\documentclass[10pt]{amsart}%
\usepackage{amsfonts}
\usepackage{amsmath}
\usepackage{amssymb}
\usepackage{graphicx}%
\usepackage{algorithm}
\usepackage[noend]{algpseudocode}

\usepackage{color}

\renewcommand{\l}{\ell}

\newcommand{\C}{\mathcal{C}}

\newcommand{\M}{\mathcal{M}}

\newcommand{\T}{\mathcal{T}}

\renewcommand{\l}{\ell}

\renewcommand{\C}{\mathcal{C}}
\newcommand{\labelings}{\mathcal{F}}
\renewcommand{\T}{\mathcal{T}}
\newcommand{\rt}{r}
\newcommand{\drf}{d_{RF}^*}
\newcommand{\match}{\mu}
\newcommand{\ub}{UB_{\mu}}

\newcommand{\size}{sz}
\newcommand{\setsize}[2]{V_{#2}({#1})}

\newcommand{\twoch}{double-cherry}
\newcommand{\twochs}{double-cherries}
\newcommand{\nbch}{chr_1}
\newcommand{\nbtwoch}{chr_2}

\newcommand{\V}{\mathcal{V}}

\newcommand{\unlabelledmast}{uMAST}

\newtheorem{theorem}{Theorem}[section]
\newtheorem{lemma}[theorem]{Lemma} 
\newtheorem{corollary}[theorem]{Corollary}

\begin{document}

\title{The complexity of comparing multiply-labelled trees by extending phylogenetic-tree metrics}
\author{Manuel Lafond}
\address{Department of Mathematics and Statistics, University of Ottawa,
  Ottawa, Canada} \email{mlafond2@uOttawa.ca}
  
  \author{Nadia El-Mabrouk}
 \address{D\'epartement d'informatique et de
  recherche op\'erationnelle, Universit\'e de Montr\'eal, Qu\'ebec,
  Canada} \email{mabrouk@iro.umontreal.ca}

  \author{Katharina T. Huber}
  \address{School of Computing Sciences, University of East Anglia, Norwich, UK}
  \email{K.Huber@uea.ac.uk}
  
  \author{Vincent Moulton}
\address{School of Computing Sciences, University of East Anglia, Norwich, UK}
  \email{v.moulton@uea.ac.uk}

%%%%%%%%%%%%%%%%%%%%%%%%%%%%

%\maketitle

\begin{abstract}
	A multilabeled tree (or MUL-tree)
          %~\cite{huber2006phylogenetic} 
	is a rooted tree in which
	every leaf is labelled by an element from some set, but in which 
	more than one leaf may be labelled by the same element of that set. In phylogenetics, such trees 
	are used in biogeographical studies, to study the evolution 
	of gene families, and also within approaches to construct phylogenetic networks. 
	A multilabelled tree in which no leaf-labels are repeated is called a
	phylogenetic tree, and one in which every label is the same	is also known as a tree-shape.
	In this paper, we consider the complexity of computing metrics on
	MUL-trees that are obtained by extending  metrics on 
	phylogenetic trees. In particular, 
	by restricting our attention to tree shapes, we show that 
	computing the metric extension on MUL-trees is NP-complete for two well-known
	metrics on phylogenetic trees, namely, the path-difference and Robinson Foulds
	distances. We also show that 
        the extension of the Robinson Foulds distance is fixed parameter tractable with respect
        to the distance parameter. 
	The path distance complexity result allows us to also answer an open problem concerning 
	the complexity of solving the quadratic assignment problem
	for two matrices that are a Robinson similarity and a Robinson dissimilarity, which
	we show to be NP-complete. We conclude by
	considering the maximum agreement subtree (MAST) 
	distance on phylogenetic trees to MUL-trees. Although its extension to 
	MUL-trees can be 
	computed in polynomial time, we show that computing 
	its natural generalization to more than two MUL-trees is NP-complete, although
	fixed-parameter tractable in the maximum degree when the number of given trees is bounded.  
\end{abstract}

\maketitle


\section{Introduction}

In phylogenetics, leaf-labelled, rooted trees are used to represent the vertical evolution of a
collection $X$ of species, genes or other units of heredity \cite{semple2003phylogenetics}. 
In this context, a multilabeled tree (or MUL-tree)~\cite{huber2006phylogenetic} is a rooted tree in which
every leaf is labelled by an element in the set $X$, but in which 
more than one leaf may be labelled by the same element of $X$. %(see Figure~\ref{example}). 
In case no leaf-labels are repeated, a MUL-tree is called a {\em phylogenetic
tree}. At the other extreme, where the 
leaves of the MUL-tree are all labelled with the same element of $X$ 
the tree is also known as a {\em tree-shape} (or simply a rooted tree).
MUL-trees appear in biogeographical studies~\cite{ganapathy2006pattern} 
where they are also known as area cladograms, in the study of the evolution 
of gene families \cite{gregg2017gene} where
multiple labels represent paralogous genes in the same genome,  and also 
within approaches to construct phylogenetic networks~\cite{huber2016folding}.
MUL-trees and related structures also appear in areas such as 
data-mining \cite{chou2005mmdt} and string-matching \cite{crochemore1997direct}.

As methods for constructing MUL-trees often results in multiple solutions 
or require searching through collections of MUL-trees, it is
important to develop systematic methods to 
compare MUL-trees \cite{ganapathy2006pattern,Huber11}. 
One way to do this is to extend metrics on phylogenetic tree-metrics to MUL-trees.
Phylogenetic trees are well understood and have been studied for several years \cite{y93}. 
Given a metric $d$ on the set of phylogenetic trees with leaf-set $X$,
we can define a metric $d^*$ on the 
set of MUL-trees with leaf-set being a multiset $M$
of size $m$ and underlying set $X$ with $|X| \ge m$ as follows. 
Given two MUL-trees $T_1$ and $T_2$ with leaf-set $M$, 
we bijectively assign labels to the leaves of $T_1$ and $T_2$ from 
the set $\{1,\dots,m\}$ 
in such a way that the two subsets of $\{1,\dots,m\}$ that are 
assigned to the leaves of $T_1$ and $T_2$ that are labelled by the same element in $X$ are equal.
This results in two phylogenetic trees $T_1^*$
and $T_2^*$ each with leaf-set  $\{1,\dots,m\}$.  We then define 
the extension $d^*(T_1,T_2)$  to be the
minimum value of $d(T_1^*,T_2^*)$ taken over all possible 
assignments of this kind. In \cite[p.1031]{Huber11} it is shown that $d^*$ 
is a metric on the set of MUL-trees with leaf-set $M$. 

As pointed out in \cite[p.1037]{Huber11} the complexity of computing 
$d^*$  is not known for some tree-metrics $d$ which are commonly used 
in phylogenetics.  In this paper, we will therefore
consider the complexity of computing $d^*$ for three well-known metrics
on phylogenetic trees: the path-difference distance $d_{path}$ \cite{y93}, the Robinson-Foulds distance
$d_{RF}$ \cite{robinson1971comparison}
and the maximum agreement subtree (MAST) distance $d_{MAST}$ \cite{goddard1994agreement}. 
Note that we do not consider the so-called nearest-neighbour interchange (NNI) distance
and related operation-based tree metrics since, in contrast to $d_{path}$, $d_{RF}$ and
$d_{MAST}$,  these are NP-complete to compute  even for phylogenetic trees (see \cite{Huber11}
and the references therein).

To determine the complexity of computing $d^*_{path}$ and $d^*_{RF}$, we 
shall restrict attention to the case where $M$ consists of a single element, i.e. we  
shall reduce the problem to tree shapes. In this special case, for two tree shapes $T_1,T_2$
with $m$ leaves,  the problem of computing $d^*$ reduces to finding 
a bijective labelling of the leaf-sets of $T_1$ and $T_2$ by the set $\{1,\dots,m\}$
which minimises the value of the metric $d$ between the resulting phylogenetic trees.
Note that other approaches have been proposed for defining metrics 
on tree shapes -- see e.g. \cite{Colijn17} and the references therein -- although
such metrics are not more generally applicable to MUL-trees. 

After presenting some preliminaries in Section~\ref{prelim}, 
we begin by considering the complexity of computing $d^*_{path}$. The path distance 
between two phylogenetic trees is essentially the sum of 
the length-differences of the paths connecting two specified leaves
in the two trees taken over every possible pair of leaves.
In Section~\ref{path} we show that computing $d^*_{path}$ is 
NP-complete (Theorem~\ref{theo:2-tree-shapes}). Our proof is based on a previous result that was used to show
that the so-called Gromov-Hausdorff
distance between metric trees is NP-hard~\cite{Agarwal15}.
%\cite{agarwal2015computing}.
Interestingly,
in Section~\ref{sec-qap} we are then able to use the fact
that computing $d^*_{path}$ is  NP-complete to solve an
open problem presented in \cite{Laurent15}. In that paper, it is
stated that the complexity is unknown for the problem of finding a
permutation which solves the quadratic assignment problem
for two matrices $P$ and $Q$, when $P$ and $Q$
are a Robinson similarity and Robinson dissimilarity, respectively (see  Section~\ref{sec-qap} 
for definitions of these terms).  Here we show that this problem is NP-complete (Theorem~\ref{theo:quadratic-assignment}).

We then turn to considering the complexity of computing $d^*_{RF}$. 
The Robinson-Foulds distance $d_{RF}$ between two 
phylogenetic trees on $X$ is essentially the size of the 
symmetric difference of the two sets of clusters induced on $X$ by the two trees.
In Section~\ref{robinson} we show that computing $d^*_{RF}$ is NP-complete (Theorem~\ref{theo:np-complete-binary}),
even for two binary tree shapes.
% (that is, tree shapes in which 
%the root node has degree 2, and all internal nodes have degree 3).
However, we shall also 
show that there is a fixed-parameter tractable algorithm for 
computing $d^*_{RF}$ (Theorem~\ref{fptrf}).  
We shall present the proof for NP-completeness for the non-binary case; as
the argument is quite long and technical, we present the proof 
for the binary case in an appendix. 

In Section~\ref{mast} we consider the MAST distance $d_{MAST}$  
between two phylogenetic trees, which is given by the
size of the leaf-set of a maximum agreement subtree 
of the two trees.
% (i.e. a phylogenetic 
%subtree that is contained in both trees having the maximum number of leaves). 
Interestingly, by results in \cite{ganapathy2006pattern},
a maximum agreement sub-MUL-tree  for two MUL-trees can be 
computed in polynomial time (even for two MUL-trees
with different size leaf-sets), from which it follows that $d^*_{MAST}$ can
be computed in polynomial time \cite{Huber11}.
Motivated by this fact, we consider the related problem
where the aim is to find a maximal agreement sub-MUL-tree 
for 3 or more MUL-trees. Note that this 
is closely related to the largest common subtree problem \cite{akutsu2000approximation}.
By reducing again to tree shapes, in Section~\ref{mast} we show that this
more general problem is NP-complete for three tree shapes when the degree of
the input trees is unbounded (Theorem~\ref{theo:mast-is-hard}).  However, we also show that the problem
is fixed parameter tractable with respect to the maximum degree if the
number of trees is constant (Theorem~\ref{fptumast}).  We conclude 
by stating some open problems in the final section.

\section{Preliminaries}\label{prelim}

All graphs considered in this paper are {\em simple},
that is, they do not contain multiple edges or loops.
We denote the edge set of a graph $G$ by $E(G)$
and its vertex set by $V(G)$.  If $G$ is a tree, we will call its vertices \emph{nodes}.
A map $\omega:E(G)\to \mathbb{R}_{\geq 0}$ is called an
{\em edge-weighting} for $G$. An {\em edge-weighted}
graph is a pair $(G,w)$ where $G$ is a graph and
$\omega$ is an edge-weighting for $G$.

Let $T$ be a {\em rooted} tree, that is, a tree with
%a unique node of indegree zero, called
%the {\em root} of $T$ and denoted by $\rt(T)$, and all edges
%directed away from  $\rt(T)$.
a distinguished node $\rt(T)$ called its \emph{root}.
We denote by $L(T)$ the set of leaves of $T$ and by $I(T)$ its set of 
internal nodes. We also view an isolated node as a rooted tree. For $u\in V(T)$
we denote by $T(u)$ the sub-tree of $T$ rooted at $u$.
The \emph{size} $\size(u) $ of a node $u \in V(T)$
is $|L(T(u))|$. The {\em size of $T$}, denoted $\size(T)$,
 is the size of $r(T)$. Note that $\size(T)=|L(T)| $. 
%We will assume that trees do not have vertices of degree $2$, except possibly the root.
 A tree is \emph{binary} if its root has degree $2$,
 and every other internal node has degree $3$.

A node $u\in V(T) $ is a {\em descendant} of a node $v\in V(T)$ 
if $v$ is on the path between 
$u$ and $\rt(T)$, and in this case $v$ is an {\em ancestor} of $u$.
Note that any node $u$ is both a descendant and an ancestor of itself.
If $v$ is an ancestor of $u$ and $u \neq v$, then $v$ is called a 
\emph{proper ancestor} of $u$ and $u$ is a $\emph{proper descendant}$ of $v$.
If $u$ is a descendant of $v$ and $uv \in E(T)$, then $u$ is a \emph{child} of $v$.
Two nodes $u$ and $v$ are \emph{incomparable} if none is an ancestor of 
the other.
A {\em common ancestor} of a subset
$L'\subseteq L(T)$ is a node $v\in V(T)$ that lies, for all
$x\in L'$, on the path from $r(T)$ to $x$. The {\em last common
ancestor} of $L'$ is the unique common ancestor $v\in V(T)$ of $L'$ 
such that no proper descendant of $v$ is also 
a common ancestor of $L'$.

Suppose $X$ is a finite non-empty set.
 A {\em rooted phylogenetic tree on $X$} is a pair 
$\T= (T, \phi)$ where $T$ is a
rooted tree for which every internal node has at least $2$ children, 
and $\phi$ is a bijective mapping from $L(T)$ into $X$ 
assigning each leaf a label\footnote{Note that this definition
  of a phylogenetic tree on $X$ is different from the one
  given in e.\,g.\,\cite{semple2003phylogenetics} in that the roles of the
  domain and co-domain of $\phi$ are swapped.}. 
We may call $\phi$ a \emph{leaf assignment} of $T$.
In case the knowledge of $X$
is of no relevance to the discussion then we refer to a rooted
phylogenetic tree on $X$ simply as a phylogenetic tree. 
For $T$ a rooted tree with $|L(T)|=|X|$ we denote
by $\labelings(T) = \{(T, \phi) : \phi \mbox{ is a bijection from }
L(T) \mbox{ to } X \}$ the set of all possible rooted phylogenetic 
trees on $X$ given by bijectively 
labeling the leaves of $T$ by the elements of $X$. 
To improve clarity of our arguments and distinguish rooted trees from phylogenetic trees, we 
refer to a rooted tree as a  
{\em tree shape} in order to emphasize that their leaves are unlabelled.

%In case $X$ has precisely one element, $x$ say,
%then  we refer to the pair $\T= (T, \phi)$ where $T$ is as
%above and $\phi$ is a mapping from $L(T)$ into $\{x\}$ 
%that assigns to each leaf of $T$ the label $x$ as
%a {\em rooted tree shape on $x$} -- see also \cite{Huber11,SS03}. 
%To help improve clarity
%we will simply refer to a rooted tree shape on some element
%$x$ as a {\em tree shape} if the knowledge of $x$ is of
%no relevance to the discussion. Thus the notions of a
%rooted tree and a rooted tree shape coincide.
%\ml{Now that $\phi$ goes from $L(T)$ into $X$, I think this paragraph is more confusing than 
%anything.  I do not see why we need trees in which each leaf maps to $x$.  
%Could we consider removing this paragraph?  The cited papers in this paragraph could be cited elsewhere.}

%We may call $(T, \phi)$ a rooted $X$-tree for short [NOTE: this is not the usual definition].
%Unless stated otherwise, we will implicitly assume that $X = [\size(T)] = [n]$.

Suppose that $n\geq 1$ and that $d$ is a metric on the set of rooted 
phylogenetic trees on $[n]= \{1, 2, \ldots, n\}$.
If $T_1$ and $T_2$ are both tree shapes with $n$ leaves,
then we put
$$
d^*(T_1, T_2) = \min_{\T_1 \in \labelings(T_1), \T_2 \in \labelings(T_2)}d(\T_1, \T_2).
$$
As remarked in~\cite[p.1031]{Huber11} for the more general case of MUL-trees,  this 
is a metric on the set of rooted tree shapes  with $n$ leaves. 

\section{Path Distance}\label{path}

Given a weighted, rooted tree
 $T$ and $i,j \in V(T)$, we let $\ell_{T}(i,j)$ denote the 
length of the shortest (undirected) path in $T$ between $i$ and $j$.
The path distance $d_{path}$ between two weighted, rooted phylogenetic
trees  $\T_1=(T_1,\phi_1)$ and $\T_2=(T_2,\phi_2)$ on $[n]$ is equal to
$$
d_{path}(\T_1,\T_2) = 
\sum_{i,j \in [n]} |\ell_{T_1}(i,j) - 
\ell_{T_2}(i,j)|.
$$
 Note that $d_{path}$
is a metric on the set of weighted, rooted phylogenetic trees on $[n]$ 
(see e.\,g.\,\cite{y93}).
We shall show that computing $d^*_{path}$ is 
NP-complete by reduction to the following problem, using  
a similar technique to that presented in
\cite[Section~3]{Agarwal15}
%\cite[Section~3]{agarwal2015computing}.
\\

\noindent UNRESTRICTED PARTITION: 
Given a multiset  $X = \{a_1,\dots,a_n\}$ of positive integers
where $n = 3m$ and $m\geq 1$, 
is it possible to partition $X$ into $m$ multisets $\{X_1, \dots , X_m\}$ 
such that all the elements in each multiset $X_i$ sum to 
%the same quantity $S = 
$(\sum_{i=1}^n a_i)/m$?\\ 

Note that this problem has been proven to be
strongly NP-complete by a reduction from 3-Partition
\cite{Agarwal15}
%\cite{{agarwal2015computing}}. 
In particular, we can assume that the size of the integers is 
polynomial in the input.

Given an integer $p >0$, we let $T_p$ denote the 
{\em rooted star shape} with $p\geq 1$ leaves
%bush with $p$ leaves, that is a the 
that is, the rooted tree shape with $p$ leaves, and 
such that every edge in $T_p$ contains the root and a leaf
(in particular, $T_p$ has $p$ edges).

Now, given an instance 
$X = \{a_1,\dots,a_n\}$ of UNRESTRICTED PARTITION
where $n=3m$ and $m\geq 1$, 
we let $T$ and $T'$ 
be the two weighted, rooted tree shapes given in
 Figure~\ref{fig:pathtrees}, where for 
$S = (\sum_{i=1}^n a_i)/m$ and all  $1 \le i \le m$,
$T_S^i$ denotes a copy of the rooted star shape $T_S$.
%,
%and all edges of the form $\{r(T^i_S),l\}$ with $l\in L(T_S^i)$,
%$1\leq i\leq m$ have weight 10.5 and all other edges have weight one.

\begin{figure*}[!t]
	\begin{center}
		\includegraphics[width= 1 \textwidth]{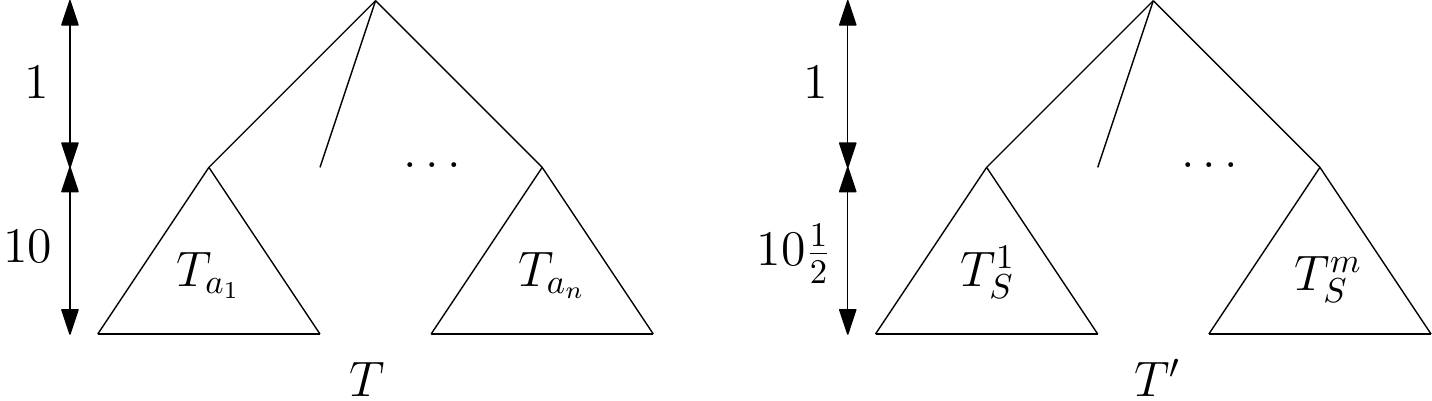}
		\caption{The weighted rooted tree shapes $T$ and $T'$.
For $T'$, every edge of the form $r(T^i_S)l$ with $l\in L(T^i_S)$
where $1\leq i\leq m$ has weight 10.5 and all remaining edges of that
tree have weight one.
Similarly, every edge of $T$ of the form $\{r(T_{a_i}), l\}$
where $l\in L(T_{a_i})$ and $1\leq i\leq n$ has weight 10 and all remaining
edges in that tree have weight one.
}\label{fig:pathtrees}
	\end{center}
\end{figure*}

Continuing with this notation we obtain:

\begin{lemma}\label{help}
Suppose we are given an instance $\{a_1,\dots,a_n\}$
of UNRESTRICTED PARTITION.
Then
$$
d_{path}^*(T,T') \ge {\sum_{i=1}^n a_i  \choose 2}
$$
with equality holding if and only if the given instance 
is a ``yes" instance.
\end{lemma}

\begin{proof}
Suppose that $\psi$ is a bijection from $L(T)$ to $L(T')$.
Let $A = {L(T)\choose 2}$ and put 
\begin{eqnarray*}
B &=& \{\{l,l'\} \in A \,:\, \{l,l'\} \subseteq L(T_{a_i}), 
\mbox{ some } 1 \le i \le n\},\\
%$$
%E = \{\{l,l'\} \in A' \,:\, \{l,l'\} \in T_S^j \mbox{ some } 1 \le j \le n\},
%$$
%and $F = A'-E$. In addition, let
B_1 &=& \{\{l,l'\} \in B \,:\, \{\psi(l),\psi(l')\} \subseteq L(T_S^j), 
\mbox{ some } 1 \le j \le m\}, \mbox{and}\\
B_2 &=& \{\{l,l'\} \in B \,:\, \{\psi(l),\psi(l')\} \not\subseteq L(T_S^j),  
\mbox{ for any } 1 \le j \le m\}.
\end{eqnarray*}
Furthermore, put 
$C=A-B$ and
\begin{eqnarray*}
C_1& =& \{\{l,l'\} \in C \,:\, \{\psi(l),\psi(l')\} \subseteq L(T_S^j), 
\mbox{ some } 1 \le j \le m\}, \mbox{and} \\
C_2 &= &\{\{l,l'\} \in C \,:\, \{\psi(l),\psi(l')\} \not\subseteq L(T_S^j), 
 \mbox{ for any } 1 \le j \le m\}.
\end{eqnarray*}
Note that $|A|= {\sum_{i=1}^n a_i  \choose 2}$, 
$B$ is the disjoint union of $B_1$ and $B_2$, and
that $C$ is the disjoint union of $C_1$ and $C_2$.
%Moreover, it is straight-forward to see 
We show that $B_2 = \emptyset$ if
and only if $\psi$ corresponds to a ``yes" instance 
of UNRESTRICTED PARTITION.

Setting  $f(l,l') = |\ell_{T}(l,l') - \ell_{T'}(\psi(l),\psi(l'))  |$  for $\{l,l'\}\in A$, we have 
\begin{eqnarray*}
\sum_{\{l,l'\} \in A } f(l,l')
& = & \sum_{\{l,l'\} \in B_1} f(l,l') + \sum_{\{l,l'\} \in B_2}f(l,l') + 
\sum_{\{l,l'\} \in C_1}f(l,l') + \sum_{\{l,l'\} \in C_2} f(l,l')\\
& = & |B_1| + 3|B_2| + |C_1|+ |C_2|\\
&=& (|B_1| + |B_2| + |C_1|+ |C_2|  ) + 2|B_2|\\
&=& |A|+2|B_2|.
\end{eqnarray*}

Now note that any pair of bijective labellings of the
leaf-sets of the tree shapes $T$ and $T'$ by the set $[mS]$
(giving rooted phylogenetic trees $\T$ and $\T'$ on $[mS]$, respectively)
gives rise to a bijection between $L(T)$ and $L(T')$.
Moreover, all bijections between $L(T)$ and $L(T')$ 
can arise in this way, and if such a bijection is the map $\psi$ 
given above, then the path-distance between
$\T$ and $\T'$ is $|A|+2|B_2|$.

Since $d^*_{path}(T, T')$ is given by taking the 
minimum over all pairs of bijective labellings of the
leaf-sets of the tree shapes $T$ and $T'$ by $[mS]$,
the lemma now follows immediately.
\qed
\end{proof}

Using the Lemma~\ref{help}, we now prove:

\begin{theorem}\label{theo:2-tree-shapes}
	For two tree shapes $T_1$ and $T_2$ of the same size, 
computing $d^*_{path}(T_1,T_2)$  
	is NP-complete, even in the case that both of the tree shapes 
are binary.
\end{theorem}

\begin{proof}
	The problem of computing $d^*_{path}$ is easily seen to be in NP, as a leaf assignment of $T_1$ and $T_2$ 
	can serve as a certificate from which the
        distance can be computed in polynomial time.
        The NP-hardness of the problem follows immediately
        from Lemma~\ref{help}, and therefore, 
	computing $d^*_{path}$ is NP-complete.

We now sketch the proof of the last statement. 
Given an instance $\{a_1,\dots,a_n\}$ with $n=3m$ and $m\geq 1$  
 of UNRESTRICTED PARTITION, 
we first turn the tree shapes used in Lemma~\ref{help} into binary tree
shapes by resolving the nodes of outdegree 3 or more in 
both tree shapes through inserting a caterpillar (tree)
shape of height $\epsilon$,
as indicated in Figure~\ref{fig:blowup} (a caterpillar shape
is a binary tree in which each internal node has exactly one child
that is a leaf, 
except for one single internal node which has two leaf children).
Since at most two unresolved nodes 
can lie on a directed path from the root of $T$ to a leaf of $T$ 
or from the root of $T'$ to a leaf of $T'$ it follows that 
we obtain two weighted binary tree shapes $T$ and $T'$ of 
 height $10+ 2\epsilon$ and $10.5 + 2\epsilon$, respectively, 
for some small $\epsilon >0$. 

Now, using the same type of argument as in Lemma~\ref{help},
it can be  seen that for a bijection $\psi$ between $L(T)$ and $L(T')$, we have 
$$
\sum_{ \{l,l'\} \in A} | \ell_{T}(l,l') - \ell_{T'}(\psi(l),\psi(l'))  |   =  |A|+2|B_2| + \epsilon g(a_1,\dots,a_n)
$$
where $A$ and $B_2$ are defined just as in the proof of Lemma~\ref{help},
and $g$ is some function of $a_1,\dots,a_n$.
Hence, by considering all bijections from $L(T)$ to $L(T')$, 
we can take $\epsilon$ sufficiently small so that
the instance $\{a_1,\dots,a_n\}$ corresponds to a ``yes" instance 
of UNRESTRICTED PARTITION if and only if 
$$
|d_{path}^*(T,T') - {\sum_{i=1}^n a_i  \choose 2}| < 1.
$$
The last statement of the theorem now follows immediately.
\qed
\end{proof}

\begin{figure*}[!t]
	\begin{center}
		\includegraphics[width= 1 \textwidth]{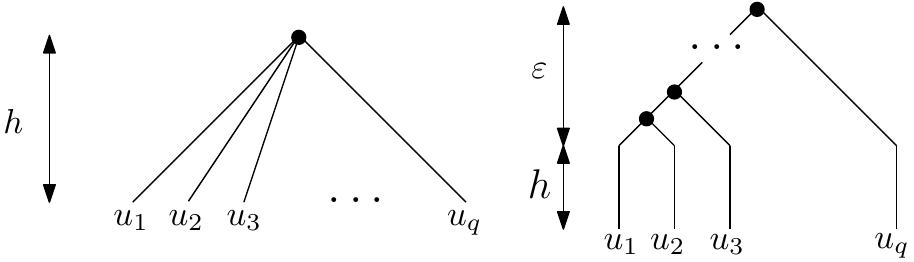}
		\caption{The replacement of an unresolved node of height $h$
by a caterpillar tree shape of height $h+\epsilon$ employed in the proof of 
Theorem~\ref{theo:2-tree-shapes}.
}\label{fig:blowup}
	\end{center}
\end{figure*}

\section{The Quadratic assignment problem for Robinson matrices}\label{sec-qap}

Given two  $N \times N$ symmetric matrices $P,Q$ with $N\geq 1$,
the  QUADRATIC ASSIGNMENT problem is to find a permutation $\psi$
of $[N]$ which minimizes
\begin{equation} \label{qap}
\sum_{i,j = 1}^N P_{ij} Q_{\psi(i)\psi(j)}.
\end{equation}

Here for a matrix $M$, $M_{ij}$ denotes the entry at row $i$ and column $j$.
For $N\geq 1$ a   symmetric $N \times N$ matrix $P$ 
is called a {\em Robinson similarity matrix} if its entries decrease 
monotonically in the rows and the columns when moving away from the main diagonal, i.e., if
$$
P_{ik} \le min\{P_{ij},P_{jk}\} \mbox{ for all } 1 \le i  \le j \le k \le N.
$$
Similarly, an $N \times N$ symmetric matrix $Q$ 
is called a {\em Robinson dissimilarity matrix} if its entries increase 
monotonically in the rows and the columns when moving away from the main diagonal.

In \cite{Laurent15}, it is stated that 
the complexity of finding a permutation which 
solves the QUADRATIC ASSIGNMENT 
problem for two symmetric $N\times N$ matrices $P$ and $Q$
when $P$ is a Robinson similarity and $Q$ 
is a Robinson dissimilarity, is not known. Here we shall show that this problem is NP-complete.
% Laurent's paper:
% The quadratic assignment problem is easy for Robinsonian matrices with Toeplitz structure
% Monique Laurent and Matteo Seminaroti

First, suppose we are given an instance $\{a_1,\dots,a_n\}$ 
of UNRESTRICTED PARTITION where  $n=3m$ and $m\geq 1$. 
For $T$ and $T'$ the weighted rooted 
tree shapes depicted in Figure~\ref{fig:pathtrees},   
define $N \times N$ matrices $P(T)$ and $Q(T')$ as follows.
For $P(T)$, label the leaves of $T$ by $1$ up to
 $N= \sum_{i=1}^n a_i$ from 
left to right and, for all $i,j \in [N]$, set $P(T)_{ij} = -l_T(i,j)$.
For $Q(T')$, also label the leaves of $T'$ by $1$ up to
 $N$ from 
left to right and, for all $i,j \in [N]$, set $Q(T')_{ij}=  l_{T'}(i,j)$.
It is straight-forward to see that $-P(T)$ is a Robinson similarity matrix
and that $Q(T')$ is a Robinson dissimilarity matrix.
We also remark that it is easy to see that a permutation $\psi$ of $[N]$
minimizes the quantity in Expression (\ref{qap}) 
for $P=P(T)$ and $Q=Q(T')$  if and only if $\psi$ 
minimizes
$$
\sum_{i,j \in [N]} (\ell_T(i,j) - \ell_{T'}(\psi(i),\psi(j)))^2.
$$

Continuing with this notation, we obtain

\begin{theorem}\label{theo:quadratic-assignment}
	The QUADRATIC ASSIGNMENT problem for $P$ and $Q$ 
is NP-complete for $P$ the Robinson similarity $P(T)$
	and $Q$ the Robinson dissimilarity $Q(T')$.
\end{theorem}
\begin{proof}
We claim that 
given an instance $\{a_1,\dots,a_n\}$
of UNRESTRICTED PARTITION where $n=3m$ 
and $m\geq 1$ and any permutation $\psi$ of $[N]$
where $N= \sum_{i=1}^n a_i$,
we have 
$$
\sum_{i,j \in [N]} (\ell_T(i,j) - \ell_{T'}(\psi(i),\psi(j)))^2 
\ge {N  \choose 2},
$$
with equality holding if and only if $\psi$ corresponds to a  
``yes" instance. The proof of the theorem then follows
immediately from the remark preceding  
Theorem~\ref{theo:quadratic-assignment}.

The proof of the claim is very similar to that of Lemma~\ref{help}
and so we only give a sketch proof. Suppose 
$\psi$ is a permutation of $[N]$. Then $\psi$ is a bijection
between $L(T)=[N]$ and $L(T')=[N]$. Define the 
sets $A= {L(T) \choose 2}$, $B_1$, $B_2$, $C_1$, and $C_2$  
in an analogous way to the sets defined in the proof of Lemma~\ref{help}.
Setting $f(l,l') = (l_T(i,j) - l_{T'}(\psi(i),\psi(j))^2$ for $\{l,l'\} \in A$, 
it follows that
$$
\sum_{\{l,l'\} \in A} f(l,l') = |B_1| + 9|B_2| + |C_1| + |C_2| = |A| + 8|B_2|.
$$ 
The proof of the claim now easily follows using a similar 
argument to that used in the last part of the proof of Lemma~\ref{help}.
\qed
\end{proof}

\section{Robinson Foulds Distance}\label{robinson}

Suppose $\T = (T, \phi)$ is a rooted phylogenetic tree on $X$.
For a node $u$ in $\T$, we denote by
$C_{\T}(u) = \{ \phi(l) : l \in L(T(u))\}$
the \emph{cluster} of $\T$
associated with $u$, that is, the set of labels assigned 
by $\phi$ to the leaves of $T(u)$. In case $\T$ is clear
from the context, we may also write $C(u)$, for short.
The set of clusters of $\T$ is $\C(\T) = \{C(u) : u \in V(T) \}$.
For $u\in V(\T)$  we call the cluster $C(u)$ 
\emph{trivial} if $u$ is either the root or a leaf 
of $T$, and \emph{non-trivial} otherwise.
The \emph{Robinson Foulds (RF) 
distance} $d_{RF}(\T_1, \T_2)$ between two rooted 
phylogenetic trees $\T_1$ and $\T_2$ on $X$
is the cardinality of the symmetric difference of 
$\C(\T_1)$ and $\C(\T_2)$, i.e.
$$
d_{RF}(\T_1, \T_2) = |\C(\T_1) \setminus \C(\T_2)| + |\C(\T_2) \setminus \C(\T_1)|$$

Note that trivial clusters never contribute towards the RF distance
between two rooted phylogenetic trees.
Also observe that if both $\T_1$ and $\T_2$  are binary, then 
$d_{RF}(\T_1, \T_2) = 2 |\C(\T_1) \setminus \C(\T_2)|$.
Given two tree shapes $T_1$ and $T_2$, the $\drf$ distance asks for a leaf assignment of $T_1$
and $T_2$ that minimizes the $d_{RF}$ distance.
We will show that computing $\drf(T_1, T_2)$ is NP-complete, 
but it is fixed-parameter tractable with respect to that distance.
Beforehand, we establish a useful connection between $\drf$ and \emph{cluster matchings}.

\subsection{Cluster matchings and the Robinson Foulds distance}

A \emph{cluster matching} $\M$ between two tree shapes $T_1$ and $T_2$
is a set of pairs $(v_1, v_2)\in V(T_1)\times V(T_2)$,
such that each node of $V(T_1) \cup V(T_2)$ appears in at most one pair of $\M$.
We say that a node $v \in V(T_1) \cup V(T_2)$ is \emph{matched in $\M$} 
if there exists 
a pair in $\M$ that contains $v$, and  \emph{unmatched in $\M$} otherwise.
We may simply say that $v$ is \emph{matched} (or \emph{unmatched}) 
if $\M$ is clear from the context.
Also, we say that $\M$ is a \emph{consistent} cluster matching 
between $T_1$ and $T_2$
if the two following properties hold:

\begin{enumerate}
\item[(M1)]
$(v_1, v_2) \in \M$ implies $\size(v_1) = \size(v_2)$;

\item[(M2)]
for any two pairs $(v_1, v_2), (v_1', v_2') \in \M$, 
$v_1$ and $v_1'$ are incomparable if and only if $v_2$ and $v_2'$ are 
incomparable.
\end{enumerate}
%
%Observe that the empty set is vaccuously a 
%consistent cluster matching between $T_1$ and $T_2$.
Note that, if $\M$ is a consistent cluster matching
and $(v_1, v_2), (v_1', v_2') \in \M$ then
$v_2$ is an ancestor of $v_2'$ if and only if $v_1$ is an ancestor of $v_1'$.

We say that a consistent cluster matching $\M$
between two tree shapes $T_1$ and $T_2$ is {\em maximum}
if for any consistent cluster matching $\M'$ between $T_1$
and $T_2$, we have $|\M| \geq |\M'|$.
Also, we denote by $\match(T_1, T_2)$ the cardinality of a 
maximum consistent cluster matching between $T_1$ and $T_2$.

To illustrate these definitions assume that $T_1$ and $T_2$ are
two rooted tree shapes with leaf sets $\{v_1,v_1'\}$ and $\{v_2,v_2'\}$,
respectively. Then
the set $\M=\{(v_1,v_2), (v_1',v_2')\}$
is a consistent cluster matching for $T_1$ and $T_2$. However
$\M$ is not maximum since the cluster matching $\M'=\M\cup
\{(r(T_1),r(T_2))\}$ is also a consistent cluster matching
for $T_1$ and $T_2$. Clearly $\match(T_1, T_2)=3$.

We next establish that in a maximum consistent cluster matching
for two tree shapes $T_1$ and $T_2$ of the same size we can always
assume that every leaf of $T_1$ is matched
with a leaf of $T_2$.

\begin{lemma}\label{lem:match-all-leaves}
Let $T_1$ and $T_2$ be two tree shapes of the same size, 
and let $\M$ be a consistent cluster matching between $T_1$ 
and $T_2$.  Suppose that at least one leaf of $T_1$ is unmatched in $\M$.
Then there exist $l_1 \in L(T_1)$ and $l_2 \in L(T_2)$ such that 
$\M \cup \{(l_1, l_2)\}$ is a consistent cluster matching.
\end{lemma}

\begin{proof}
We first consider the case that $(r(T_1), r(T_2)) \in \M$. Then
we can choose some $v \in V(T_1) \cup V(T_2)$ such that $v$ is matched in $\M$, 
$v$ has a descendant leaf $l$ that is unmatched in $\M$ 
and $\size(v)$ is as small as possible (the assumption that $(r(T_1), r(T_2)) \in \M$ guarantees 
that such a node $v$ exists).
Suppose without loss of generality that $v \in V(T_1)$, and let $v' \in V(T_2)$ 
be such that $(v, v') \in \M$.  Because of the remark following the
definition of a consistent cluster matching, it follows that 
$v'$ must have an unmatched descendant leaf $l'$.  By the choice of $v$, 
no node on the $l - v$ path is matched in $\M$ except $v$, and no node on the 
$l' - v'$ path is matched in $\M$ except $v'$.  Since $\M$ is a consistent
cluster matching and Property~(M1) is satisfied for
the pair $p=(l, l')$  and Property~(M2) is satisfied for $p$ and
any pair of $\M$ it follows that
$\M \cup \{(l, l')\}$ is a consistent cluster matching for $T_1$ and $T_2$.

So assume that $(r(T_1), r(T_2)) \not\in \M$. Then it is straight-forward
to see that $\M'=\M\cup\{(r(T_1), r(T_2))\}$ is a consistent cluster
matching for $T_1$ and $T_2$. In view of the previous case, there
exist some leaf $l_1$ in $T_1$ and some leaf $l_2$ in $T_2$
such that $\M' \cup \{(l, l')\}$ is  a consistent cluster matching
for $T_1$ and $T_2$. But then
$\M' \cup \{(l, l')\}-\{(r(T_1), r(T_2)) \}=\M \cup \{(l, l')\}$
is also a consistent cluster matching for $T_1$ and $T_2$.
\qed
\end{proof}

The following result allows us to reformulate the problem of computing 
the $\drf$ distance between two tree shapes of the same size
in terms of maximum consistent 
cluster matchings.

\begin{lemma}\label{lem:cluster-matching}
Suppose $T_1$ and $T_2$ are two tree shapes  of the same size. Then 
$\drf(T_1, T_2) = |V(T_1)| + |V(T_2)| - 2\match(T_1, T_2)$.

Furthermore, given a consistent cluster matching $\M$ of cardinality $\match(T_1, T_2)$, 
let $\phi_1$ and $\phi_2$ be leaf assignments of $T_1$ and $T_2$, respectively, 
such that for all leaves $l_1 \in L(T_1), l_2 \in L(T_2)$, the property that $(l_1, l_2) \in \M$
implies that $\phi_1(l_1) = \phi_2(l_2)$.
Then $d_{RF}((T_1, \phi_1), (T_2, \phi_2)) = \drf(T_1, T_2)$.
\end{lemma}

\begin{proof}
We first show that $\drf(T_1, T_2) 
\geq |V(T_1)| + |V(T_2)| - 2\match(T_1, T_2)$.
Let $\phi_1:L(T_1) \to X$ and $\phi_2:L(T_2) \to X$
denote two leaf assignments of $T_1$ and $T_2$, respectively,
such that $d_{RF}(\T_1, \T_2) = \drf(T_1, T_2)$
where $\T_1 = (T_1, \phi_1)$ and $\T_2 = (T_2, \phi_2)$.
Let $\M = \{(v_1, v_2) \in V(T_1)\times V(T_2): C_{\T_1}(v_1) = C_{\T_2}(v_2)\}$.
Clearly, $\M$ is a consistent cluster matching.
Moreover, $d_{RF}(\T_1, \T_2)$ is the number of nodes that are 
unmatched in $\M$ in each tree shape,
 and so $\drf(T_1, T_2) = d_{RF}(\T_1, \T_2) = 
|V(T_1)| + |V(T_2)| - 2|\M| \geq |V(T_1)| + |V(T_2)| - 2\match(T_1, T_2)$.

Conversely, let $\M$ be a consistent 
cluster matching of cardinality $\match(T_1, T_2)$.  
By Lemma~\ref{lem:match-all-leaves}, every leaf $l \in L(T_1) \cup L(T_2)$ 
is matched in $\M$.  Let $\phi_1$ and $\phi_2$ be leaf assignments of 
$T_1$ and $T_2$, respectively, such that for every $l_1 \in L(T_1)$ and $l_2 \in L(T_2)$, 
$\phi_1(l_1) = \phi_2(l_2)$ if and only if $(l_1, l_2) \in \M$.  
Now, let $(v_1, v_2) \in \M$.  By Property (M2), every $l \in L(T_1(v_1))$ is matched with 
a leaf $l' \in L(T_2(v_2))$.  This implies that $C_{(T_1, \phi_1)}(v_1) = C_{(T_2, \phi_2)}(v_2)$.
That is, if a node $u$ of $T_1$ or $T_2$ is matched in $\M$, then the $u$ cluster does not contribute 
to the $\drf$ distance between $T_1$ and $T_2$ under $\phi_1$ and $\phi_2$.
Thus, $\drf(T_1, T_2) \leq d_{RF}((T_1, \phi_1), (T_2, \phi_2)) \leq  |\{u \in V(T_1) \cup V(T_2) : u$ is not matched in $\M\}|$.
The number of unmatched nodes is $V(T_1) + V(T_2) - 2|\M|$, proving the first claim of the Lemma.

The second statement of the lemma follows immediately, as we have just shown that $d_{RF}((T_1, \phi_1), (T_2, \phi_2)) \leq V(T_1) + V(T_2) - \match(T_1, T_2)$, 
and that $V(T_1) + V(T_2) - \match(T_1, T_2) = \drf(T_1, T_2)$. 
\qed
\end{proof}

Note that in view of the arguments used in the proof of
Lemma~\ref{lem:cluster-matching}, the problem of
minimizing the $\drf$ distance 
between two tree shapes is equivalent 
to finding a maximum consistent cluster matching between them.
As we shall see it will sometimes be more convenient to formulate 
the $\drf$ minimization problem
in this latter form.

We also observe that if both tree shapes $T_1$ and $T_2$
are binary then $|V(T_1)| = |V(T_2)| = 2\size(T_1) - 1$
(see e.g. \cite{semple2003phylogenetics}). Combined with Lemma~\ref{lem:cluster-matching},
we obtain the following result

\begin{corollary}
Suppose $T_1$ and $T_2$ are two binary tree shapes of the same size. Then 
$\drf(T_1, T_2) = 4\size(T_1) - 2 - 2\match(T_1, T_2))$.
\end{corollary}

%In [REF], it is shown that $\drf$ satisfies the triangle inequality.  
Furthermore, as $\drf$ is a metric  and therefore satisfies the
triangle inequality we also obtain the following result via
substitution.
%the following can then be obtained directly from 
%$\drf(T_1, T_3) \leq \drf(T_1, T_2) + \drf(T_2, T_3)$ 
%by substitution.

\begin{corollary} \label{cor:triangle}
If $T_1, T_2, T_3$ are three tree shapes of the same size, 
then $|V(T_2)| + \match(T_1, T_3) \geq \match(T_1, T_2) + \match(T_2, T_3)$.
\end{corollary}

\subsection{NP-completeness for the non-binary case of Robinson Foulds}

Using a reduction from the DOMINATING SET problem,
we establish in this section
that computing $\drf(T_1, T_2)$ is NP-complete, even 
in trees of height at most $3$. In the  DOMINATING SET problem, we are given a 
connected graph $G = (V, E)$ and an integer $k \geq 1$, and ask 
if $G$ has a dominating set of size at most $k$, where a {\em dominating set}
 is a subset $D$ of nodes of a graph $H$ such that
for every node $v \in V(H)$ we have that either $v \in D$,
or $v \notin D$ and there exists $u \in D$ such that $u$ and
$v$ are adjacent. 

Let $(G, k)$ be an instance of DOMINATING SET, with $n = |V(G)|\geq 3$.
We next outline the construction of two tree shapes $T_1$ and $T_2$
from $G$ the details of which we give below
(see Figure~\ref{fig:non-binary-reduction} for an illustration
of the case $n=3$).
Assume that  $(v_1, \ldots, v_n)$ is an (arbitrary) ordering  on
$V(G)$. Then tree shape $T_1$ is built from two types of tree shapes. For
$i \in [n]$ these are tree shapes rooted at a node $w_i$
which corresponds to $v_i$, and tree shapes rooted 
at a node $d_{i,j}$ which corresponds to the edge $v_iv_j$ of $G$.
As we shall see, $\size(w_i) = n^2 + i$ holds for each $i \in [n]$ and $\size(d_{i,j} = j + 1$ 
for every $v_iv_j \in E(G)$.

To start with, the root of the tree shape $T_1$ has $n$ 
children $w_1, \ldots, w_n$. 
Then apply the following procedure for each $i \in [n].$
Let $v_{j_1},\ldots, v_{j_l}$
be the neighbors of $v_i$ in $G$ (noting that $1 \leq l \leq n - 1$). %must hold as a vertex of
%$G$ cannot be a neighbor of itself.
Then we add to $w_i$ exactly $l + 1$ children $d_{i, j_1}, d_{i, j_2}, 
\ldots, d_{i, j_l}$ and 
$d_{i, i}$, each respectively of size $j_1 + 1, j_2 + 1, \ldots, j_l + 1$
and $i + 1$.
These sizes are achieved by inserting $j_p + 1$ leaves as children of 
$d_{i,j_p}$ for each $p \in [l]$,  and $i + 1$ leaves as children of $d_{i,i}$.
Observe that so far, $w_i$ has size at most
$\sum_{p=1}^l (j_p + 1) +i+1\leq (n - 1)n + i + 1 < n^2 + i$.
We add leaf children to $w_i$ until $w_i$ has size $n^2 + i$ 
(as in Figure~\ref{fig:non-binary-reduction}).
% Note that if
%$v_i$ was an isoloated vertex, $w_i$ would have only one in $G$ this final step ensures
%that $T_1$ is indeed a tree shape.

As for $T_2$, the root $r(T_2)$ has $2n$ children $d_1', \ldots, d_n', w_1', \ldots, w_n'$. 
For each $i \in [n]$, we insert $i + 1$ leaves as children of $d_i'$ and $n^2 + i$ leaves as 
children of $w_i'$.
Hence $\size(d'_i) = i + 1$ and $\size(w_i') = n^2 + i$ for each $i \in [n]$.
%These 
%are the star tree shapes $D_1', \ldots, D'_n$ of sizes 
%$2, 3, \ldots, n + 1$ respectively, 
%plus the star tree shapes $W_1', \ldots, W_n'$ of sizes $n^2 + 1, n^2 + 2, \ldots, n^2 + n$.
To finish the construction of $T_1$ and $T_2$, we add leaf children to the
smallest of $\rt(T_1)$ 
or $r(T_2)$ until $\size(T_1) = \size(T_2)$
(Figure~\ref{fig:non-binary-reduction} shows leaf insertions under $r(T_1)$).

Notice that in a consistent cluster matching for $T_1$ and $T_2$, the 
$d'_j$ nodes of $T_2$ can only be matched with the nodes of $T_1$
of the form  $d_{i,j}$, 
and the $w'_i$ nodes of $T_2$ with the nodes of $T_1$ of the form  $w_i$.
Moreover, since all the $w'_i$ and $d'_j$ nodes are incomparable, 
all the nodes of $T_1$ that are matched with the $w'_i$ and $d'_j$ 
nodes must also be 
incomparable.

\begin{figure*}[h]
	\begin{center}
		\includegraphics[width= 0.95 \textwidth]{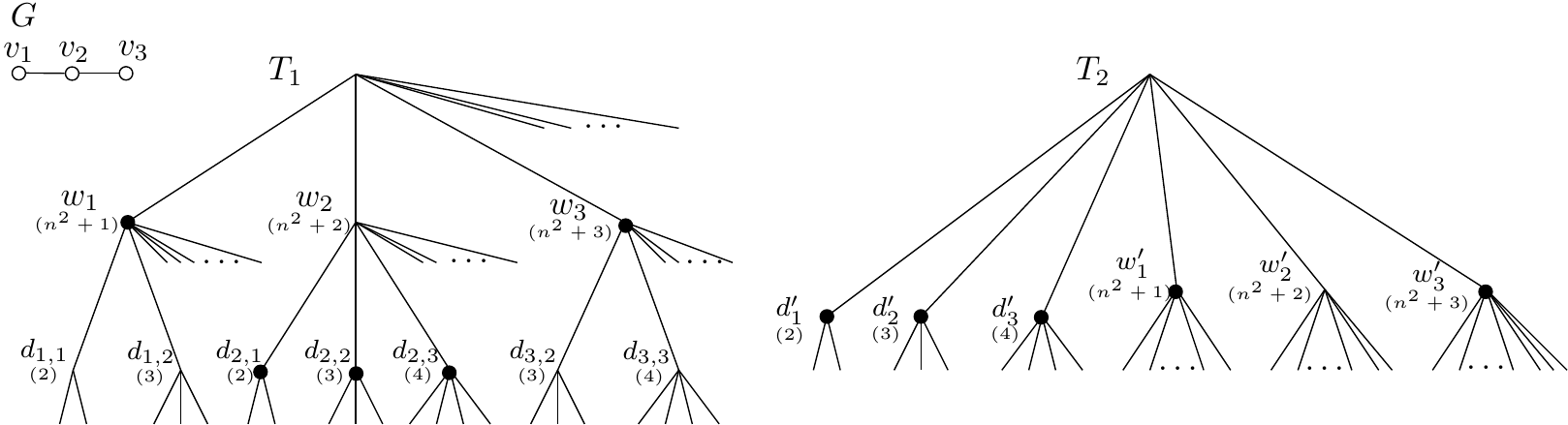}
		\caption{An example of the construction of $T_1$
                  and $T_2$ from $G$ in the case that $n = 3$. The
                  ordering of the vertices is $(v_1,v_2,v_3)$. 
%To improve clarity,
 %                 we write $d_{i,p}$ rather than $d_{i,j_p}$, for all
 %                 $i\in[n]$ and all $p\in[l]$. 
		The size
                  of a node is shown in parentheses so, for example,
                  $\size(w_1)=(n^2+1)$. 
                  %Edges in $T_1$ incident with $w_i$
                  %but not of the form $\{w_i,d_{i,j}\}$ or $\{r(T_1), w_i\}$
                  %end in a  leaf of $T_1$. 
	       %For this we assumed
                  %$\size(T_1)<\size(T_2)$. 
	      %\ml{This might be confusing since the two trees are supposed to have the same size.  The remark in the construction text should be enough.}
	      A maximum cluster matching
                  which corresponds to the dominating set is $D = \{v_2\}$
                  depicted by black nodes (the vertices $w_2$ and $w_2'$
                  are unmatched).  See the proof of Theorem~\ref{theo:np-complete-height-3} for details.
}\label{fig:non-binary-reduction}
	\end{center}
\end{figure*}

\begin{theorem} \label{theo:np-complete-height-3}
  For two tree shapes $T_1$ and $T_2$ of the same size,
  finding $\drf(T_1, T_2)$ is NP-complete
	even if both trees have height at most $3$.
\end{theorem}

\begin{proof}
  The problem is clearly in NP, since every bijective labelling
  of the leaf sets of $T_1$ and $T_2$ in terms of a set
  of size $\size(T_1)$ can
  serve as a certificate of the distance and can be verified easily.
	For hardness, we show that, given a DOMINATING SET 
instance $(G, k)$, for the two tree shapes 
	$T_1$ and $T_2$ constructed above, $G$
 has a dominating set of size at most $k$
	if and only if there exists a consistent cluster matching of size
	at least $\size(T_1) + 1 + 2n - k$ 
where $n$ is again $|V(G)|$.

Using the same notation as in that construction, including our arbitrary ordering $(v_1, \ldots, v_n)$ on $V(G)$, 
assume first that $D = \{v_{i_1}, \ldots, v_{i_k}\}$ 
is a dominating set of $G$ of size $k$.

	We construct a consistent cluster matching $\M$ 
        between $T_1$ and $T_2$ by first matching every $d'_j$ node of $T_2$
        to a node in $T_1$ as follows. Start with $\M=\emptyset$.
	For each $v_j\in V(G)$, let $v_i \in D$ be a vertex dominating
        $v_j$ (with $i = j$ if $v_j \in D$). 
%\kh{I assume that the required ordering of the vertices
%          of $G$ is induced by the dominating sets. Maybe we should comment on
%        this.}
%\ml{I recalled the fact that we are using the same ordering as in the construction.}
        Then	we add the pair  $(d_{i,j}, d'_j)$ to $\M$ (note that $d_{i,j}$ must exist, as either $i = j$ or $v_iv_j \in E(G)$). Since the empty set
        is vacuously a consistent cluster matching between $T_1$ and $T_2$,
        it follows that the resulting set $\M$ is also a consistent
        cluster matching between $T_1$ and $T_2$
        since all the $d_j$ nodes are incomparable in $T_2$, 
        and all the nodes of $T_1$ of the form $d_{i, j}$ are also
        incomparable in $T_1$. Since, by construction,
        $r(T_2)$ has $n$ children of the form $d_j'$ it follows that so far, $|\M|=n$.

        Observe that since we have only matched children of $w_i$ nodes corresponding to the $D$ vertices, 
exactly $n - k$ of the child nodes of $r(T_1)$ are
{\em free}, that is, do not have a descendant that is matched in $\M$.
For all $i\in [n]$, we match the $w'_i$ child nodes of $r(T_2)$ with the
$w_i$ child nodes of  $r(T_1)$
that are free, i.e. we add $\{(w_i, w'_i) : v_i \notin D\}$ to $\M$.  
The resulting set $\M$ is clearly a consistent cluster matching
for $T_1$ and $T_2$ and contains $2n - k$ pairs.  
%Since the height of $T_2$ is two it follows that 
%these $2n-k$ pairs in $\M$ induce a canonical consistent cluster
%matching $\M'\subseteq L(T_1)\times L(T_2)$ between $T_1$ and $T_2$.
%\ml{I DONT GET THIS}.
Using Lemma~\ref{lem:match-all-leaves}, 
we may add $\size(T_1)$ leaf pairs to $\M$ as well as the
pair  $(r(T_1), r(T_2))$ to $\M$, which results in a consistent cluster 
matching of size $\size(T_1) + 1 + 2n - k$.
	
	Conversely,  assume that  $\M$ is  a consistent 
cluster matching between $T_1$ and $T_2$  
	that contains at least $\size(T_1) + 1 + 2n - k$ pairs.  
        We show first that we may assume that in $\M$, every $d'_j$
        node is matched with a node in $T_1$ as 
	otherwise we may transform $\M$ as follows.
	If $d'_j$ is not matched with some node in $T_1$
        for some $j\in [n]$, it must be that no node 
of $T_1$ of the form $d_{i,j}$ is available.
	This happens only if for every $w_i$ node of $T_1$ having
a child  $d_{i, j}$, we have that $w_i$ is already matched (with $w'_i$).
	Furthermore, because $w_i$ is matched, no non-leaf child $d_{i, p}$ of $w'_i$  
can be matched with a node in $T_2$, because $w'_i$ does not have children that are available to match.
%the height of $T_2$ is two and the
%height of $T_1$ is three.  
Thus by replacing 
	$(w_i, w'_i) \in \M$ by $(d_{i, j}, d'_j)$ (and rematching the 
leaves appropriately using Lemma~\ref{lem:match-all-leaves}), 
	we obtain a consistent cluster matching of the same size. 
	We may repeat this process for every $j$ until every
        $d'_j$ node is matched with a node in $T_1$, each time updating
        $\M$.
	
	Now, let $D = \{v_i \in V(G): (d_{i, j}, d'_j) \in \M 
        \mbox{ and } i,j \in [n] \}$.
	%for every $j \in [n]$, let $\alpha(D'_j)$ denote the unique index $i$ 
	%such that $(D_{i, j}, D'_j) \in \M$.  Then let $X = \{v_i : i = \alpha(D'_j)$ for some $j \in [n]\}$.
	For each $j \in [n]$, since $d_{i, j}$ exists in $T_1$ if and only 
if $\{v_i,v_j\}$ is an edge of $G$ or $i = j$, it follows that
	$D$ must be a dominating set (because every $d'_j$ node is matched).  It remains only to show that 
$|D| \leq k$.
	To see this, observe that if a node of $T_1$ of the form $d_{i, j}$ 
        is matched with a node of $T_2$, then this node must be $d_j'$.
        In this case, the parent $w_{i}$ of $d_{i,j}$ cannot be matched, 
since $d_j'$ does not have an ancestor of size $\size(w_i) = n^2 + i$.
Thus if $|D| > k$ held then there would be fewer than $n - k$ of 
the $w_i$ nodes of $T_1$ that 
can be matched with nodes in $T_2$.
It follows that $|\M| < \size(T_1) + 1 + 2n - k$,
a contradiction. 
	Therefore, $D$ is a dominating set of size at most $k$, as desired.
	\qed
\end{proof}

Note that the arguments in the proof of 
Theorem~\ref{theo:np-complete-height-3} can be 
extended to show that computing 
$\drf$ between two \emph{binary} trees shapes 
is NP-complete (Theorem~\ref{theo:np-complete-binary}).  
However, the proof requires a careful handling of the details, 
and is significantly more technical.  
We redirect the interested reader to the Appendix.

\begin{theorem}\label{theo:np-complete-binary}
	For two \emph{binary} tree shapes $T_1$ and $T_2$ of the same size, finding $\drf(T_1, T_2)$ is NP-complete.
\end{theorem}

%, and explain here the essence of the modifications 
%that are necessary to adapt the non-binary proof.
%First, in the above proof, we conveniently have limited the 
%``matchable'' nodes to the $D_{i, j}$ and $W_i$ subtrees, and each of these trees had only one choice in $T_2$.  
%However, a binary tree has much more internal nodes than the above trees, 
%and limiting the ``matchable'' nodes in a similar manner is more difficult, as we need to avoid 
%creating undesired matching possibilities accidentally.
%We circumvent this problem by making (almost) all the nodes of $T_1$ of odd size, 
%and (almost) all the nodes of $T_2$ of even size - with the exception of those nodes that we want to be 
%matchable.  The second difficulty
%is that binary trees inevitably have cherries (usually many), which are automatically matchable.
%Thus when constructing $T_1$ and $T_2$, it is important to keep track of how many
%cherries are present in each subtree so as not to allow an alternative solution 
%to the one we wish to enforce for our reduction.  This is handled by bookkeeping the cherries 
%in $T_1$ and $T_2$ (we also need to keep track of the subtrees on four leaves).

\subsection{Fixed-parameter tractability of the RF distance}

We show that deciding for two tree shapes $T_1$ and $T_2$ of the same
  size and a
non-negative integer $k$
if $\drf(T_1, T_2) \leq k$ can be done 
in time $O(2^{k \log (2k + 1)} n^3)$, where $n = \size(T_1) = \size(T_2)$.  
Thus computing $\drf$ for $T_1$ and $T_2$ is in FPT
with respect to parameter $k$.
This result holds independently of the maximum degree of $T_1$ and $T_2$.
We need some more notation first.

If $T$ is a tree shape and $u\in V(T)-L(T)$ is a child of $r(T)$,
%with $v_1, \ldots, v_l$, $l\geq 2$, the children of $u$, 
we denote by $T - u$ the tree shape  obtained  from $T$ by
collapsing the edge $\{r(T),u\}$ (i.e. removing $u$ and passing its children to $r(T)$).
%by letting the parent of $v_i$ become $r(T)$ for each $i \in [l]$,
%then removing $u$ from $T$.
We denote by $T - T(u)$ the tree shape 
obtained by removing the $T(u)$ subtree, ignoring the root 
if it is of degree one. 
More precisely, if $r(T)$ has at least $3$ children, $T - T(u)$ is the result of deleting $u$, all of its descendants,
and all their incident edges.
If $r(T)$ has $2$ children, $T - T(u)$ is obtained as in the previous case, 
plus deleting $r(T)$ and its incident edge.

For a given integer $h\geq 1$, 
let $u_1, u_2, \ldots, u_l$ be the children of $r(T)$ of size $h$.
Write $u_i \simeq u_j$ if $T(u_i)$ is isomorphic to $T(u_j)$.
The $\simeq$ relation is clearly an equivalence relation.
We denote by $ch_{\simeq}(r(T), h)$ the subset of children of $r(T)$ of size $h$ obtained 
by choosing exactly one child (arbitrarily) for each equivalence class of this relation.

On a high level, our FPT algorithm is a top down recursive search 
tree algorithm in which the parameter $k$ indicates how many nodes are allowed 
 to contribute to the RF distance in $V(T_1) \cup V(T_2)$.
The idea is that if a child $u$ of $r(T_1)$ has size larger
than any node of $T_2$, then $u$ can never be matched in any consistent cluster matching,
and will therefore contribute to the $\drf$ distance.
We may thus remove $u$, decrease $k$ by 
$1$  and recurs.  The same holds if a node $u'$ of $T_2$
has the same property.  If otherwise $u$ can be matched to 
some nodes of $T_2$, then we simply try every possibility
(including \emph{not} matching $u$), recursing in each case.  The list of 
possibilities is given by $ch_{\simeq}(r(T_2), \size(u))$,
which has bounded cardinality, provided that we eliminate pairs of isomorphic sub-tree shapes as we will show.
We start with  the following Lemma.

\begin{lemma} \label{lem:t1-minus-u}
Let $T_1$ and $T_2$ be two tree shapes of the same size, 
and let $u$ be a child of $r(T_1)$ such that $u \notin L(T_1)$.  
Then $\drf(T_1, T_2) \leq \drf(T_1 - u, T_2) + 1$.
\end{lemma}

\begin{proof}
  The lemma follows from the fact that in any leaf assignment, the set of clusters
  induced by $T_1-u$ is the
set of clusters induced by $T_1$ with the cluster induced
by $u$ removed.
%Consider optimal leaf label assignments $\phi_1$ of $T_1 - u$ and $\phi_2$ of $T_2$,
%i.e. such that $d_{RF}((T_1 - u, \phi_1), (T_2, \phi_2)) = \drf(T_1 - u, T_2)$. Since
%the set of clusters induced by $T_1-u$ is the
%set of clusters induced by $T_1$ with the cluster induced
%by $u$ removed,  $d_{RF}((T_1, \phi_1), (T_2, \phi_2)) \leq d_{RF}((T_1 - u, \phi_1), (T_2, \phi_2)) + 1$.
\qed
\end{proof}

The next result ensures that for computing $\drf(T_1, T_2) $
we can essentially ignore pairs $(A_1,A_2)$ of isomorphic sub-tree shapes
where the root of $A_i$ is a child of $T_i$ for $i \in \{1,2\}$, 
as the nodes of $A_1$ and $A_2$ can be made to have no contribution towards 
the $\drf$ distance. 
Although this is quite an intuitive statement, the proof is rather technical and for the sake of readability, we defer it to the end of this section.

% The idea behind the algorithm is as follows.
% The parameter $k$ indicates how many nodes are allowed 
% to contribute to the RF distance in $V(T_1) \cup V(T_2)$.
% We look at all the children of $r(T_1)$ and $r(T_2)$ of maximum size.
% If $T_1$ has a child $u_1$ of maximum size, and every child of $T_2$
% has size strictly smaller than $\size(u_1)$, then clearly $u_1$ 
% can never be matched and must contribute to the RF distance.
% If $T_2$ has a child $u_2$ such that $T(u_1)$ is isomorphic to $T(u_2)$, 
% we show that some optimal solution 
% assigns the descendants of $u_1$ to those of $u_2$
% so that the two subtrees do not contribute to the RF distance.
% If no such $u_2$ exists, we list every child $b_1, \ldots, b_r$ of $T_2$ of size $\size{u_1}$
% and try mapping $u_1$ to each $b_i$.  
% In fact not every $b_i$ needs to be tested, as it suffices 
% to restrict the $b_i$'s to $$ch_{\simeq}(r(T_2), \size(u_1))$$.

\begin{lemma} \label{lem:removeiso}
Let $T_1$ and $T_2$ be two tree shapes of the same size, and 
suppose that $r(T_1)$ has a child $u$ and $r(T_2)$ 
has a child $u'$ such that $T_1(u)$ and $T_2(u')$ 
are isomorphic tree shapes.
Let $\sigma: V(T_1(u)) \to V(T_2(u'))$  be the underlying bijection
for an isomorphism between
$T_1(u)$ and $T_2(u')$.

Then there exists leaf assignments $\phi_1$ of $T_1$ and 
$\phi_2$ of $T_2$ such that the two following properties hold:

\begin{itemize}
\item
for every leaf $l \in T_1(u)$, we have $\phi_1(l) = \phi_2(\sigma(l))$;
\item
 $d_{RF}((T_1, \phi_1), (T_2, \phi_2)) = \drf(\T_1, \T_2)$.
\end{itemize}

%Then $\drf(T_1, T_2) = \drf(T_1 - T_1(u), T_2 - T_2(u'))$.
\end{lemma}

Algorithm~\ref{algo:rfdist} describes our FPT procedure in detail.  
In a first phase, we ``epurate'' $T_1$ and $T_2$ by eliminating 
children of $r(T_1)$ of size greater than those of $r(T_2)$, as they are nodes that are certain to contribute to the RF distance. 
The same applies for the large size children of $r(T_2)$.  Each removal decrease $k$, the number of nodes that are allowed 
to contribute to the RF distance, by $1$.
If there are no such nodes to eliminate, then the roots of $T_1$ and $T_2$ both have a child of maximum size $p$.
We then eliminate all pairs of isomorphic sub-tree shapes of $T_1$ and $T_2$ of size $p$ - 
with a special case to handle when one tree has a degree $2$ root and not the other.
We repeat this process until $T_1$ and $T_2$ have non-isomorphic children of size $p$. 
This is where we have to try every possible matching between these children.

\begin{algorithm}
\caption{Algorithm for the RF distance.  $T_1$ and $T_2$ are two trees of the same size, and $k$ is the maximum
allowed RF distance between $T_1$ and $T_2$.  
The Algorithm returns $\drf(T_1, T_2)$ if it is at most $k$, or $\infty$ if it is above $k$.
We assume that 
both trees shapes are to be assigned the same set of labels.}\label{algo:rfdist}
\begin{algorithmic}[1]
\Procedure{rfdist}{$T_1, T_2, k$}
\State \algorithmicif~$k < 0$
		\algorithmicthen~Return $\infty$

\State $done \gets False$
\While{$done = False$}
	\State \algorithmicif~$T_1$ and $T_2$ are isomorphic
		\algorithmicthen~Return $0$ \label{algo:rf:returniso}
	\State Let $u_1$ (resp. $u_2$) be a child of $r(T_1)$ (resp. of $r(T_2)$) of maximum size

	\If{$\size(u_1) < \size(u_2)$}
		   \State Return \Call{rfdist}{$T_1, T_2 - u_2, k - 1$} + 1		\label{algo:rf:removeu2}
	\ElsIf{$\size(u_2) < \size(u_1)$}
		   \State Return \Call{rfdist}{$T_1 - u_1, T_2, k - 1$} + 1	\label{algo:rf:removeu1}
	\ElsIf{$\size(u) = \size(v)$}
		  \State $p \gets \size(u_1)$

	\While{$r(T_1)$ has a child $u_1$ of size $p$ and $r(T_2)$ has a child $u_2$ of size $p$ \label{algo:rf:removeiso} \\
	\hspace{1.8cm} such that $T_1(u_1)$ and $T_2(u_2)$ are isomorphic}
		\If{$p = \size(T_1)/2$, $r(T_1)$ has $2$ children $\{u_1, u'\}$ and \\
		\hspace{2.23cm} $r(T_2)$ has at least $3$ children}
			\State Return \Call{rfdist}{$T_1 - u', T_2, k - 1$} + 1		\label{algo:rf:spcase1}
		\ElsIf{$p = \size(T_1)/2$, $r(T_2)$ has $2$ children $\{u_2, u'\}$ and \\
		\hspace{2.85cm} $r(T_1)$ has at least $3$ children}
			\State Return \Call{rfdist}{$T_1, T_2 - u', k - 1$} + 1	\label{algo:rf:spcase2}
		\Else
			\State $T_1 \gets T_1 - T_1(u_1)$				\label{algo:rf:removeiso1}
			\State $T_2 \gets T_2 - T_2(u_2)$				\label{algo:rf:removeiso2}
		\EndIf
	\EndWhile

	\If{both $r(T_1)$ and $r(T_2)$ both still have at least one child of size $p$}
		\State done = True
	\EndIf
	\EndIf
\EndWhile

\State Let $\{a_1, \ldots, a_s\} = ch_{\simeq}(r(T_1), p)$
\State Let $\{b_1, \ldots, b_t\} = ch_{\simeq}(r(T_2), p)$
    
\If{$s> k$ or $t > k$}
	\State Return $\infty$  \label{algo:rf:lplusr}
\Else \label{algo:rf:matcha1}
    	
	\State $best \gets$ \Call{rfdist}{$T_1 - a_1, T_2, k - 1$}$ + 1$
        \For{$i \in [t]$}
        	\State $dist_1 \gets$ \Call{rfdist}{$T_1(a_1), T_2(b_i), k - 1$} \label{algo:rf:try-match-a1}
	\State $dist_2 \gets$ \Call{rfdist}{$T_1 - T_1(a_1), T_2 - T_2(b_i), k - 1$}    \label{algo:rf:try-match-a1:second}
            	\State \algorithmicif~$dist_1 + dist_2 < best$ \algorithmicthen~$best \gets dist_1 + dist_2$  
        \EndFor
        
        \State Return $best$
    	
\EndIf

\EndProcedure
\end{algorithmic}
\end{algorithm}

\begin{theorem}\label{fptrf}
For a given pair of tree shapes $T_1$ and $T_2$ of size $n$ and an integer $k$, 
Algorithm~\ref{algo:rfdist} correctly decides if $\drf(T_1, T_2) \leq k$ and runs in time 
$O(2^{k \log (2k + 1)} n^3)$.
\end{theorem}

\begin{proof}
The algorithm creates a search tree of recursive calls, 
where the root is the initial call and the leaves are the terminal cases.
We show by induction over the height of a node in this search tree that if $\drf(T_1, T_2) \leq k$, then Algorithm~\ref{algo:rfdist} returns $\drf(T_1, T_2)$, and otherwise it returns $\infty$.
This is clearly true for the terminal cases when $k = -1$
or when $T_1$ and $T_2$ are isomorphic.  

Otherwise, note that if at any point the node $u_2$ is removed on line~\ref{algo:rf:removeu2}, 
it is because no node of $T_1$ has size $\size(u_2)$.  
Thus, $u_2$ will inevitably contribute to the $\drf$ distance, and by Lemma~\ref{lem:t1-minus-u}, 
$\drf(T_1, T_2) = \drf(T_1, T_2 - u_2) + 1$.  This justifies removing $u_2$ and reducing $k$ by $1$.
The same argument applies when removing $u_1$ on line~\ref{algo:rf:removeu1}.

Consider now the special case that occurs on line~\ref{algo:rf:spcase1}.
Here, $r(T_1)$ has exactly two children $u_1, u'$ both of size $p$ and $r(T_2)$ has one child of size $p$, 
and other children which must all have size strictly smaller than $p$.
Since $T_1(u_1)$ and $T_2(u_2)$ are isomorphic, 
by Lemma~\ref{lem:removeiso} we may assume that there is a leaf assignment of $T_1$ and $T_2$ that minimizes $\drf(T_1, T_2)$ and that matches 
the leaves (and clusters) of $T_1(u_1)$ and $T_2(u_2)$.  It follows that the $u'$ cluster under this assignment must contribute to the RF distance 
(as $u_2$ is the only node of $T_2$ of its size, and it already matched with $u_1$).
This justifies line~\ref{algo:rf:spcase1} and, by symmetry, line~\ref{algo:rf:spcase2}.

Suppose that the algorithm reached lines~\ref{algo:rf:removeiso1}-\ref{algo:rf:removeiso2}.
We wish to show that $\drf(T_1, T_2)$ remains unchanged after executing these two lines, 
i.e. that $\drf(T_1, T_2) = \drf(T_1 - T_1(u_1), T_2 - T_2(u_2))$.
Note that this latter statement is not true in general - we merely show it true for the $T_1$ and $T_2$ 
at this point in the algorithm.
By Lemma~\ref{lem:removeiso}, we may assume that there is a leaf assignment of $T_1$ and $T_2$ that minimizes $\drf(T_1, T_2)$ in which 
none of the nodes of $V(T_1(u_1)) \cup V(T_2(u_2))$ contributes to the $\drf(T_1, T_2)$ distance, as their leaves are matched together according to some isomorphism.
It is therefore easy to see that $\drf(T_1, T_2) = \drf(T_1 - T_1(u_1), T_2 - T_2(u_2))$ if both $r(T_1)$ 
and $r(T_2)$ have at least $3$ children, or if both have $2$ children.
This condition must hold at this point, since its negation was verified by the two special cases just above.
Therefore, we have introduced no error by removing the two isomorphic sub-tree shapes.

It follows that the algorithm is correct if it returns when it is inside the main \textbf{while} loop.
Assume that the algorithm exits the loop without returning.

Let $A = \{a_1, \ldots, a_s\}$ and $B = \{b_1, \ldots, b_t\}$ be the sets identified by the algorithm after this phase.
Suppose that $s > k$ or $t > k$ and line~\ref{algo:rf:lplusr} is executed.  To establish that 
returning $\infty$ by our algorithm is justified, it suffices by symmetry to consider the
case that $s > k$.  First note that due to line~\ref{algo:rf:removeiso}, for any $a \in A$, 
no sub-tree shape of $T_2$ can be isomorphic to $T_1(a)$.  Also observe that
in any leaf assignment $\phi_1$ of $T_1$ and $\phi_2$ of $T_2$, the cluster $C_{(T_1, \phi_1)}(a)$
is either preserved (i.e. in $\C(T_2, \phi_2)$), 
or not. 
If $C_{(T_1, \phi_1)}(a)$ is not preserved, then it contributes to $\drf(T_1, T_2)$. 
If it is preserved, then there exists some $b \in B$, such that $C_{(T_1, \phi_1)}(a) = C_{(T_2, \phi_2)}(b)$. 
Since $T_1(a)$ and $T_2(b)$ cannot be isomorphic, there must exist a node in
one of them  that contributes to $\drf(T_1, T_2)$.
In both cases, each $a_i \in A$ implies the existence of a distinct node of $V(T_1) \cup V(T_2)$ that 
contributes to the $\drf$ distance, implying $\drf(T_1, T_2) > k$ under the assumption that $s > k$.
The argument is the same if $t > k$.  Thus  line~\ref{algo:rf:lplusr} is justified.

Suppose instead that the `else' on line~\ref{algo:rf:matcha1}
is entered, and consider the $a_1$ node.  
Then again, under any leaf assignments $\phi_1$ of $T_1$ and $\phi_2$ of $T_2$, 
the cluster induced by $a_1$ is either not preserved, or it is preserved.  In the latter case, there is some child $b$ of $r(T_2)$ of size $p$ such that $T_2(b)$ is isomorphic 
to $T_2(b_j)$, where $b_j \in B$, 
such that  $C_{(T_1, \phi_1)}(a_1) = C_{(T_2, \phi_2)}(b)$
(this is because there are no other nodes of size $p = \size(a_1)$).
The algorithm tests all these cases (it is clearly only necessary to try matching $a_1$ to only 
one representative per equivalence class of $ch_{\simeq}(r(T_2), p)$).
To justify reducing to $k - 1$ in the recursive call of line~\ref{algo:rf:try-match-a1} and~\ref{algo:rf:try-match-a1:second}, observe that both the pairs $\{T_1(a_1), T_2(b_i)\}$
and $\{T_1 - T_1(a_1), T_2 - T_2(b_i)\}$ have $\drf$ distance at least $1$, since they are non-isomorphic.
Thus if one of the tree pairs has distance strictly more than $k - 1$, $\drf(T_1, T_2) > k$.
The correctness of the algorithm then follows by induction.

As for the complexity, first assume that we initially  
computed, for each pair of nodes $u \in V(T_1)$ and $v \in v(T_2)$, whether $T_1(u)$ and $T_2(v)$ are isomorphic.
This can be done in time $O(n^3)$ (see e.g.~\cite{buss1997alogtime}). 
The search tree created by the algorithm has depth at most $k$ and maximum degree $1 + 2k$ (one recursive call for $a_1$ unmatched, plus $2$ for each $i \leq t \leq k$), and it is straightforward to see that one pass through the algorithm takes $O(n^3)$ 
time (the first while loop is iterated at most $O(n)$ times since each iteration eliminates at least one vertex, and the inner while loop iterates over $O(n^2)$ pairs of vertices).  The complexity is therefore $O( (2k + 1)^k n^3 + n^3) = O(2^{k \log (2k + 1)} n^3 )$.
\qed
\end{proof}

We now present the proof of Lemma~\ref{lem:removeiso}.

\begin{proof}[of Lemma~\ref{lem:removeiso}]
%Let $\sigma$ be the underlying bijection
%$\sigma: V(T_1(u)) \to V(T_2(u'))$ for the isomorphism between
%$T_1(u)$ and $T_2(u')$.
To prove the lemma, we claim that there exists a maximum consistent cluster matching $\M$ 
between $T_1$ and $T_2$ such that $(x, \sigma(x)) \in \M$
for every $x \in V(T_1(u))$.  
By Lemma~\ref{lem:cluster-matching}, this is sufficient to prove our lemma, 
as $\M$ can be turned into leaf assignments $\phi_1$ and $\phi_2$ satisfying $\phi_1(l) = \phi_2(\sigma(l))$ for 
every leaf $l \in L(T_1(u))$.

To prove our claim, let $\M$ be a maximum
consistent cluster matching that maximizes 
the number of vertices $x$ of $T_1(u)$ such that 
$(x, \sigma(x)) \in \M$.  
%Without loss of generality
%we may assume that $\M$ is not maximum as otherwise the
%lemma holds. 
%Note that we may also assume without loss
%of generality that we have $(l, \sigma(l)) \in \M$ 
%for all leaves $l \in L(T_1(u))$ as otherwise
%we can modify $\M$ so that the lemma holds without decreasing its 
%cardinality. 
Assume for contradiction that $\M$ does not satisfy our claim.
Note that since $\M$ is maximum, by Lemma~\ref{lem:match-all-leaves} 
we may assume that every leaf in $T_1$ is matched in $\M$
with a leaf in $T_2$.
If we have $(l, \sigma(l)) \in \M$ for every leaf $l \in L(T_1(u))$, 
then it is not hard to see that $\M$ can be modified to satisfy
$(x, \sigma(x)) \in \M$ for all $x \in V(T_1(u))$
without decreasing its cardinality.  So assume that this is not
the case.

Let $(x, y) \in \M$ be chosen so that 
(i) either $x \in V(T_1(u))$ and $y \neq \sigma(x)$
or $y \in V(T_2(u'))$ and $x \neq \sigma^{-1}(y)$, 
and (ii) that $\size(x) = \size(y)$ is maximum among all possible choices.  
Note that $(x, y)$ exists since, in particular, some leaf $l$
of $T_1(u)$ is not matched with $\sigma(l)$.
Assume w.l.o.g. that the former case holds, i.e. $x \in V(T_1(u))$
but $y \neq \sigma(x)$.
Let $x' = \sigma(x)$ and let $\M_{x'} = \{(z, z') \in \M : z'$ is a
descendant of $x'\}$.  
Call $(z, z') \in \M_{x'}$ 
\emph{maximal} if, for every $(w, w') \in \M_{x'}$, $w'$ is not a
proper ancestor of $z'$.  Let $(z_1, z_1'), \ldots, (z_k, z'_k)$, $k\geq 1$,
be the maximal elements of $\M_{x'}$.
In particular, if $x'$ is matched then there is only one maximal element.
Note that $\sum_{i \in [k]} \size(z_i) = \sum_{i \in [k]}\size(z'_i) =
\size(x') = \size(x) = \size(y)$.

We next argue that by matching the nodes of $T_1(x)$ with those 
of $T_2(x')$, and the nodes of $T_1(z_1), \ldots, T_1(z_k)$ with
those of $T_2(y)$, we obtain another matching $\M'$ between
$T_1$ and $T_2$ in which more nodes $v$ of $T_1(u)$ are
matched with their images $\sigma(v)$ in $T_2(u')$ which
contradicts the choice of $\M$.

We start with claiming
%that any proper ancestor of a node in
%$\{x, y, z_1, \ldots, z_k, z'_1, \ldots, z'_k\}$ is not matched in $\M$,
% except $r(T_1)$ or $r(T_2)$.  
that  except for $r(T_1)$ and $r(T_2)$, no proper ancestor of a node in
$\{x, y, z_1, \ldots, z_k, z'_1, \ldots, z'_k\}$ is matched in $\M$.
For $x$, assume that  a non-root proper ancestor $w$ of $x$ is matched.
Then $w \in V(T_1(u))$ as $u$ is a child of $r(T_1)$.
By the choice of $x$, we must have $(w, \sigma(w)) \in \M$.
Let $M_1 = \{(v, v') \in \M : v$ is a descendant of $w\}$ and 
$M_2 = \{(v, \sigma(v)) : v \in V(T_1(w))\}$.  
Then as $T_1(w)$ and $T_2(\sigma(w))$ are isomorphic, 
$\M' := (\M \setminus M_1) \cup M_2$ is a consistent cluster matching, and it is maximum since $|M_2| \geq |M_1|$.
Moreover, by our choice of $\M$, 
we must then have $M_1 = M_2$, as otherwise $\M'$ would have more pairs of the form $(v, \sigma(v))$ than $\M$.
In particular, $(x, \sigma(x)) \in M_2 = M_1 \subseteq \M$, a contradiction to our choice of $x$.

For $y$, if it has a matched non-root proper
ancestor, then by consistency this ancestor would be matched 
to a non-root proper ancestor of $x$, which cannot be the
case as we just argued.

To see that no proper ancestor of $z'_i$, $i\in [k]$, can
be matched except $r(T_2)$,
and, therefore, that no proper ancestor of $z_i$, $i\in [k]$, can be matched
except $r(T_1)$ it suffices to show
 in view of the maximality of the $z'_i$'s that
no proper ancestor $w'$ of $x'$ except $r(T_2)$ can be matched.
If $w'$ is matched with $\sigma^{-1}(w')$ then since $\M$ is a
consistent cluster matching, and $w'$ is an ancestor of $x'$ it
follows that $x=\sigma^{-1}(x')$ has a matched ancestor
which is impossible. If $w'$ is not matched with $\sigma^{-1}(w')$
then we obtain a contradiction to the choice of $(x, y)$ since
$\size(w')>\size(x')=\size(x)$. This completes the proof of the claim.

These facts allows us to rearrange the matchings of the nodes in the set
$\{x, y, z_1, \ldots, z_k, z'_1, \ldots, z'_k\}$ 
without breaking consistency, as we now describe.
Consider the cluster matching $\M'$ obtained 
by removing from $\M$ any pair containing a 
descendant of a node in $\{x, y, z_1, \ldots, z_k, z'_1, \ldots, z'_k\}$ (noting that no pair of the form $(v, \sigma(v))$ was removed),
and then adding to $\M$ the following pairs.  Let $T_z = T_1(z_1)$ if $k = 1$, and otherwise
$T_z$ is the tree obtained 
by joining the roots of $T_1(z_1), \ldots, T_1(z_k)$ under a
common parent, creating an ``artificial root''.
Let $\delta = 1$ if $k > 1$ and $\delta = 0$ if $k = 1$ ($\delta$
indicates whether $T_z$ contains an artificial root).
Add to $\M'$ all the pairs $(v, \sigma(v))$ for each $v \in V(T_1(x))$,
and add all the pairs 
of a maximum cluster matching between $T_z$ and $T_2(y)$,
excluding $(r(T_z), y)$ if $k > 1$.  Then $\M'$ is a consistent
cluster matching between $T_1$ and $T_2$, 
since no non-root ancestor of the nodes in
$\{x, y, z_1, \ldots, z_k, z'_1, \ldots, z'_k\}$
is matched in $\M$ (and hence property (M2) is preserved).  Moreover, $\M'$ has more pairs of 
the form $(v, \sigma(v))$ than $\M$ since $y\not=\sigma(x)$ and
$(x,y)\in \M$.

It remains to show that $\M'$ is a {\em maximum} consistent cluster
matching between $T_1$ and $T_2$, which proves the
lemma as this contradicts
our choice of $\M$.
Put $T_x = T_1(x), T_{x'} = T_2(x')$ and $T_y = T_2(y)$.
%To see that $\M'$ is a maximum consistent cluster matching
%between $T_1$ and $T_2$,
Notice that 
$$
|\M'| = |\M| - \match(T_x, T_y) - (\match(T_z, T_{x'}) - \delta) + \match(T_x, T_{x'}) + (\match(T_z, T_y) - \delta).
$$

Using Corollary~\ref{cor:triangle}, we have
\begin{align*}
\match(T_x, T_{x'}) + \match(T_z, T_{y}) 
&= |V(T_x)| + \match(T_z, T_y) \\
&\geq 
\match(T_z, T_x) + \match(T_x, T_y) \\ 
&= \match(T_z, T_{x'}) + \match(T_x, T_y)
\end{align*}
where we use $T_x$ and $T_{x'}$ interchangeably since they are
isomorphic tree shapes.
This implies $|\M'| \geq |\M|$, which concludes the proof.
\qed
\end{proof}

\section{The Maximum Agreement Subtree (MAST) on tree shapes}\label{mast}

In this section, we study the Maximum Agreement Subtree (MAST) between 
two or more tree shapes. 
Unlike the previous sections, here we do not require that the given tree shapes have the same size.
This prompts some additional definitions.
For a finite set $X$, a \emph{rooted partial $X$-tree}  is a pair $\T = (T, \phi)$ where $T$ is a rooted tree shape 
and $\phi$ is an injection from $L(T)$ into $X$ (hence $\size(T) < |X|$ is possible).  
Put $\chi(\T) = \{\phi(l) : l \in L(T)\}$, i.e. $\chi(\T)$ is the set of labels that appear at the leaves of $\T$.
If $T$ is a tree shape, let $\labelings_X(T) = \{(T, \phi) : \phi$ is an injection from $L(T)$ into $X\}$
be the set of all possible rooted partial $X$-trees that have $T$ as their underlying tree shape.

Let $T$ be a rooted tree shape, let $L \subseteq L(T)$ 
and let $x$ be the last common ancestor of 
$L$ in $T$.  The \emph{restriction} of $T$ to $L$, denoted $T|_L$, 
is the rooted tree shape obtained 
from $T(x)$ by first deleting every node of $T(x)$
that does not have a descendant in $L$ and
then contracting the nodes in the resulting tree
with a single child until no such node remains.
If $\T = (T, \phi)$ is a rooted partial $X$-tree and $X' \subseteq \chi(\T)$, 
the {\em restriction} $\T|_{X'}$ of $\T$ to $X'$ is the pair $(T|_L, \phi|_L)$ 
where $L = \{l \in L(T) : \phi(l) \in X'\}$. Clearly,  $\T|_{X'}$ is a
rooted phylogenetic tree on $X'$.

Suppose $\T_1, \ldots, \T_k$, $k\geq 2$, are rooted partial $X$-trees
and $X' \subseteq \bigcap_{i \in [k]}\chi(\T_i)$. 
Then we say that the trees in 
$\tau=\{\T_1, \ldots, \T_k\}$ \emph{agree} on $X'$ if 
the induced rooted partial $X$-trees in the (multi)set
$\tau|_{X'}=\{\T_1|_{X'}, \ldots, \T_k|_{X'}\}$ are pairwise isomorphic
such that the underlying bijections preserve the elements in $X'$.
In this case, we call a tree in $\tau|_{X'}$ an \emph{agreement subtree} 
of $\tau$.  
We denote by $MAST(\tau)$ the maximum size of a subset 
$X' \subseteq X$ such that the trees in $\tau$ agree 
on $X'$, and we call a tree in $\tau|_{X'}$ a \emph{maximum agreement subtree
(MAST)} for $\tau$.  
Note that if $M$ is a MAST for $\tau$ 
then for all $\T_i \in \tau$ with underlying tree shape $T_i$,
there exists an injection $f_{i}:V(M)\to V(T_i)$ from $V(M)$
into the node set of $T_i$ such that if $x$ is a descendant of $y$ in $M$ (resp. $x$ and $y$ are incomparable in $M$), then 
$f_{i}(x)$ is a descendant of $f_{i}(y)$ in $T_i$ (resp. $f_i(x)$ and $f_i(y)$ are incomparable in $T_i$). 
Note that $f_i$ is sometimes called an {\em embedding} of $M$ into $T_i$. 
We say that the set of injections
$f_1, \ldots, f_{k}$ {\em witness} the fact that $M$ is a MAST
for $\tau$.

For a given set $\tau$ of tree shapes $T_1, \ldots, T_k$, possibly of different sizes,
we define the \emph{unlabelled MAST of $\tau$}, denoted
$\unlabelledmast(\tau) $ by putting $X = [\max_{i \in [k]} \{\size(T_i)\}]$ and 
$$
\unlabelledmast(\tau) =
 \max_{\T_1 \in \labelings_X(T_1), \ldots, \T_k \in \labelings_X(T_k)} 
MAST(\{\T_1, \ldots, \T_k\}).
$$
It is known that 
computing $\unlabelledmast(T_1, T_2)$ for two tree shapes (and even 
more generally for two MUL-trees) $T_1$ and $T_2$
can be done in quadratic time~\cite{ganapathy2006pattern}.  As noted in the introduction, it 
follows that the extension $d^*_{MAST}$  of the MAST distance 
$d_{MAST}$ on phylogenetic trees can be computing in polynomial 
time \cite[p.1033]{Huber11}.

\subsection{MAST on three tree-shapes}\label{three}

In this section, we show that the problem of computing an unlabelled MAST 
is NP-complete on three tree shapes of the same size when the 
degree of the input tree shapes is unbounded.  However, in the next section we shall also 
show that the  problem is FPT with respect to the 
maximum degree if the number of tree shapes is constant.

Our proof of NP-completeness is an adaptation of~\cite{amir1997maximum}.
We reduce from the RESTRICTED 3D-MATCHING problem, 
shown to be hard in~\cite{chlebik2003approximation}.
In this problem, we are
given an integer $k \geq 0$ and three pairwise disjoint sets 
$V_1, V_2, V_3$ each with $n\geq 2$ elements, 
and a set $E \subseteq V_1 \times V_2 \times V_3$ 
of triplets such that every $v \in V_1 \cup V_2 \cup V_3$ occurs in 
exactly $2$ triplets.  We ask
if there exists a 3D-Matching of size $k$, i.e. a subset 
$E' \subseteq E$ of size at least $k$ such that no two elements 
of $E'$ intersect (when thought of as $3$-sets). To present our reduction, we 
define a {\em caterpillar shape} to be a rooted binary tree shape 
in which every internal node has at least one leaf child.

For our reduction we next construct three rooted tree shapes $T_1$, $T_2$,
and $T_3$ of the same size as follows (Figure~\ref{fig:mast-reduction} illustrates the reduction). Put $\V := V_1 \cup V_2 \cup V_3$
and $E = \{e_1, \ldots, e_m\}$ (note that $m = 2n$). 
To each triplet $e_i$ we associate a rooted 
tree shape $U_i$ on $m + 2$ leaves as follows:
%which will be used 
%once for each element of $e_i$.
%The tree $U_i$ is 
start with a caterpillar shape
on $m + 1$ leaves, and list the leaves $l_1, l_2, \ldots, l_{m + 1}$
in non-increasing order of depth.  Obtain $U_i$ by grafting a new leaf $x_i$ on 
the edge between $l_{i + 1}$ and its parent.  Observe that in this manner, 
no two tree shapes $U_i$ and $U_j$ are isomorphic, 
but become so when removing $x_i$ and $x_j$ (and their incident
edges) from $U_i$ and $U_j$ (suppressing the resulting degree two node
in each).

To each element $v \in \V$ we then associate
a tree shape $T_v$ as follows.  
Let $e_{i_1}, e_{i_2}$ be the triplets 
containing $v$, 
% and let $x$ denote the new leaf grafted onto $U_{i_1}$
%and let $y$ denote the new leaf grafted onto $U_{i_2}$.
let $T_v$ be the tree shape obtained by taking a copy of 
$U_{i_1}$ and $U_{i_2}$, 
then joining their roots under a common parent.
Observe that the $T_v$ sub-tree shapes differ only by the placement of their $x_{i_1}$ and $x_{i_2}$ leaves.
Then for each $i \in \{1,2,3\}$, the tree shape $T_i$ is  obtained by 
joining the roots of $\{T_v : v \in V_i\}$ under a common parent.
Note that for all $i=1,2,3$, the tree shape $T_i$ has $n(2m + 4)$ leaves.

\begin{figure*}[!t]
	\begin{center}
		\includegraphics[width= 0.9 \textwidth]{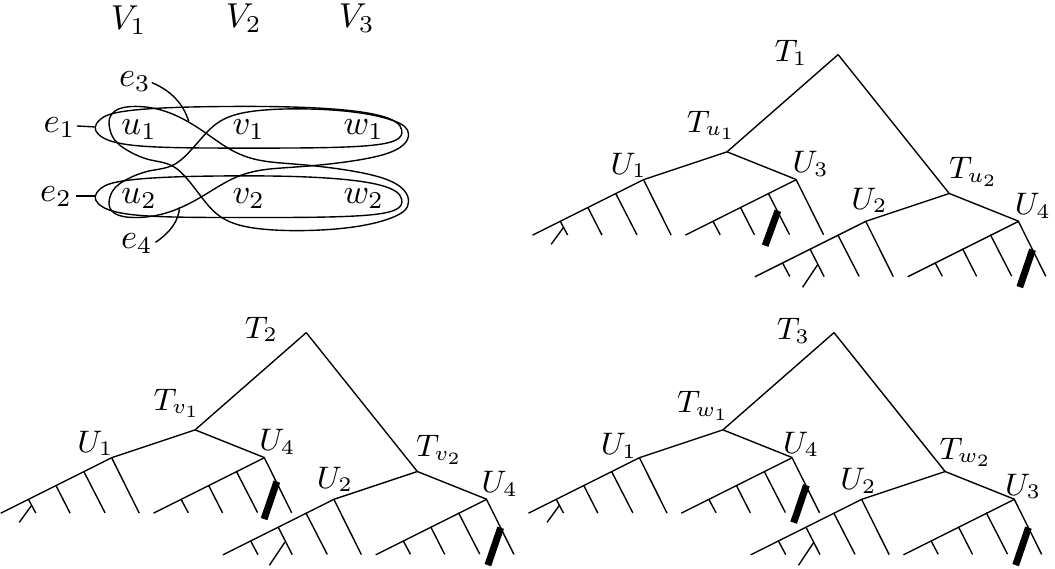}
		\caption{An example of the reduction with $n = 2$ and $m = 2n = 4$.  On the left, the RESTRICTED-3D-MATCHING instance.  The triplets are $e_1 = (u_1, v_1, w_1), e_2 = (u_2, v_2, w_2), e_3 = (u_1, v_2, w_2), e_4 = (u_2, v_1, w_1)$.  There is a 3D-matching of size $k = 2$ comprising of $e_1$ and $e_2$.  
The corresponding MAST of size $n(2m + 2) + k = 22$ constructed in the proof of hardness of Theorem~\ref{theo:mast-is-hard} can be obtained after removing the heavy edges along with their incident leaf. }\label{fig:mast-reduction}
	\end{center}
\end{figure*}

\begin{theorem}\label{theo:mast-is-hard}
For three tree shapes $T_1, T_2$ and $T_3$ of the same size, computing 
$\unlabelledmast(\{T_1, T_2, T_3\})$ is NP-complete.
\end{theorem}

\begin{proof}
To see that the problem is in NP, observe that a MAST can 
serve as a certificate, 
since subtree isomorphism can be checked in polynomial time~\cite{buss1997alogtime}.
As for hardness, let $V_1, V_2, V_3, E$ and $k$ form an instance of RESTRICTED 3D-MATCHING, 
and let $T_1, T_2$ and $T_3$ be tree shapes constructed as above.
We claim that the given RESTRICTED 3D-MATCHING instance admits a 3D-matching 
of size at least $k$
if and only if the associated tree shapes $T_1, T_2$ and $T_3$
agree on a MAST of size at least $n(2m + 2) + k$.

Assume first that 
$E' \subseteq E$ is a 3D-matching of size $k$.
The proof works by ``matching'' for $(u,v,w) \in E'$ the
 $T_u$, $T_v$ and $T_w$ sub-tree shapes 
of $T_1, T_2$ and $T_3$.  To be more precise,
for each $e_i = (u,v,w) \in E'$, we assign to each of
$T_u, T_v$ and $T_w$  the same set $Y$ of leaf labels 
so that the resulting rooted phylogenetic trees on $Y$ 
agree on a MAST of 
size $2(m + 2) - 1 = 2m + 3$ 
(note that this can easily be done since each of the tree shapes $T_u, T_v$ and $T_w$ 
contains a copy of $U_i$ as a subtree, as in e.g. the $T_{u_1}, T_{v_1}, T_{w_1}$ sub-tree shapes in Figure~\ref{fig:mast-reduction}).  
Let $F \subseteq V_1 \times V_2 \times V_3$ be
a subset of triples such that each element of $\V$  
not contained in an element of $E'$ occurs exactly once.
Note that $F = \emptyset$ is possible if $E'$ is a perfect 3D-matching, as in Figure~\ref{fig:mast-reduction}.
In that particular case, we clearly have $n = k$, and $T_1, T_2$ and $T_3$ agree on a set $X'$ of size at least
$n(2m + 3) = n(2m + 2) + n$, as desired.  Otherwise, assume $F \neq \emptyset$.
For each $(u,v,w) \in F$, assign the same set $Y'$ of
leaf labels to $T_u, T_v$ and $T_w$ 
so that the resulting rooted phylogenetic  trees on $Y'$ agree on a MAST $M'$
of size $2(m + 2) - 2 = 2m + 2$ (this can be achieved by removing the $x_i$ leaves in the $U_i$ sub-tree shapes).
%rendering them 
%labelled trees $\T_u', \T_v'$ and $\T_w'$ on $Y'$.  
%Note that this can be done so that $\{\T_u', \T_v',\T_w'\}$ 
% agree on a MAST $M'$
%of size $2(m + 2) - 2$.
Taking the disjoint union of the leaf sets of the MASTs constructed above, we obtain a bijective labelling 
of the leaf sets for the tree shapes $T_1$, $T_2$ and $T_3$
such that they agree on a  
MAST of size $k(2m + 3) + (n - k)(2m + 2) = n(2m + 2) + k$.

Conversely, assume that 
there exists a leaf assignment of the tree shapes 
$T_1$, $T_2$, and $T_3$
in terms of a set $X$ such that the resulting phylogenetic  trees 
$\T_1$, $\T_2$, and $\T_3$  on $X$ 
agree on a MAST $M$ of size $n(2m + 2) + k$.  
Then for all $i\in [3]$ there exists an 
injection $f_i: V(M) \rightarrow V(T_i)$ 
witnessing that $M$ is a MAST for $\{\T_1,\T_2,\T_3\}$.   
First note for all $i\in [3]$ that $f_i(r(M)) = r(T_i)$ since 
$r(T_i)$ is the only node of $T_i$ 
that has size at least $n(2m + 2) + k$. Note next that for each $i \in [3]$, each child of $r(M)$
must be mapped to a child of $r(T_i)$ under $f_i$. Indeed, suppose for
contradiction that there exists some $i\in [3]$ and some 
 child $z$ of $r(M)$ such that $f_i(z)$ is not a child of $r(T_i)$.  
Then $f_i(z)$ belongs to some $U_j$ sub-tree shape, and $\size(z) \leq m + 2$. 
Let $w$ be the child of $r(T_i)$ that is an ancestor of $f_i(z)$. 
Then $f_i^{-1}(w) = \emptyset$, and so $M$ has at most $m + 2$ leaves mapped to the $T_i(w)$ subtree under $f_i$.  
As the other children of $r(T_i)$ each have $2m + 4$ children, it follows that the size of $M$ can be at most 
$m + 2 + (n - 1)(2m + 4) < n(2m + 2) + k$ when 
$m = 2n$ which is impossible.
%This is strictly smaller than $n(2m + 2) + k$ whenever $n < m + 2$, a contradiction.

Now, consider a child $z$ of $r(M)$.  Let $u' = f_1(z), v' = f_2(z)$ and 
$w' = f_3(z)$.  By the above argument, $u',v'$ and $w'$ are the roots of 
$T_u$, $T_v$ and $T_w$, respectively, where
$u \in V_1, v \in V_2, w \in V_3$.
%, and thus of 
%the sub-tree shapes $T_u, T_v, T_w$. Put $\tau=\{T_u,T_v,T_w\}$.
Put $\tau=\{T_u,T_v,T_w\}$.
Then $\size(z) =\size( \unlabelledmast(\tau))$.  
It is not hard to see that the tree shapes in
$\tau$ have at most one $U_i$ sub-tree shape 
in common, which happens 
if and only if $(u,v,w) \in E$.  In this case, 
$\size(\unlabelledmast(\tau)) = 2(m + 2) + 1$ and if there is no such triplet, 
$\size(\unlabelledmast(\tau)) = 2(m + 2)$.  
Therefore, $M$ has size $n(2m + 2) + k$ only if $r(M)$ has $k$ children 
that each correspond to a triplet of $E$.  This correspondence yields a 3D-matching of size $k$.
\qed
\end{proof}

\subsection{An algorithm for few tree shapes and bounded degree}

We present an algorithm that computes $\unlabelledmast(\tau)$ of $k\geq 2$ 
tree shapes, possibly of different sizes, in time $O((2n)^k \cdot d^3 \cdot 2^{d \log d \cdot (k + 1/2)})$, 
where $d \geq 2$ is the maximum degree of a node of a tree shape in $\tau$ 
and $n \geq 2$ is the maximum size of a tree in $\tau$.  
The algorithm is  inspired by the algorithm for the labeled 
version presented in~\cite{farach1995agreement}.

Let $\tau =\{T_1, \ldots, T_k\}$ be a set of tree shapes of size at most $n$ in which each node 
of a tree shape has at  most $d$ children.
A \emph{node vector} $\vec{v} = (v_1, v_2, \ldots, v_k)$ is a sequence 
of nodes where $v_i \in V(T_i)$ for each $i\in [k]$. From now on, we
denote the $i$-th node in a node vector $\vec{v}$ for $\tau$ by $v_i$.
We write $\vec{u} \leq \vec{v}$ if $u_i$ is a descendant of $v_i$ 
for each $i \in [k]$, and $\vec{u} < \vec{v}$ if additionally 
$\vec{u} \neq \vec{v}$.
We say $\vec{u}$ is a \emph{direct predecessor} of $\vec{v}$ 
if $u_i \neq v_i$ 
for exactly one $i \in [k]$, and $u_i$ is a child $v_i$.
Thus $\vec{u}$ is obtained from $\vec{v}$ by replacing 
for some $i\in [k]$ the node  
$v_i$ by one of its children in $T_i$.
Denote by $DP(\vec{v})$ the set of direct predecessors of $\vec{v}$.  Since the 
degree of a  tree shape in $\tau$ is at most $d$, it is clear that
$DP(\vec{v})$ has at most $kd$ elements.

For a node vector $\vec{v}$, a \emph{child vector} for $\vec{v}$ is 
a node vector $\vec{v'} = (v'_1, \ldots, v'_k)$ 
where, for each $i \in [k]$, $v'_i$ is a child of $v_i$.
%an integer between 
%$1$ and the number of children of $v_i$ (intuitively, $d_i$ refers to 
%the $d_i$-th child of $v_i$).
We call two child vectors $\vec{v'}$ and $\vec{v''}$ 
of $\vec{v}$ \emph{compatible} 
if $v'_i \neq v''_i$ for every $i \in [k]$.

For a node vector $\vec{v}$ for $\tau$, denote 
by
$\unlabelledmast(\vec{v}) $ the size of the unlabelled MAST 
for the set of tree shapes obtained by restricting
each tree shape in $T_i$, $i\in[k]$  to the tree shape $T_i(v_i)$
rooted at $v_i$. Put differently,
$\unlabelledmast(\vec{v})= 
\unlabelledmast(T_1(v_1), \ldots, T_k(v_k))$.
Clearly, our goal is to obtain $\unlabelledmast(\vec{v}) $
where $\vec{v} = (r(T_1), \ldots, r(T_k))$.
We achieve this by computing 
$\unlabelledmast(\vec{v})$ for each one of the possible $O((2n)^k)$
node vectors $\vec{v}$ for $\tau$ by dynamic programming.

% The idea is simply to consider, for a $MAST$ $M$ for $\vec{v}$, 
% the vector $\vec{q} = (f_1(r(M)), \ldots, f_k(r(M)))$.  
% There are two possibilities: 
% either $\vec{q} < \vec{v}$ or $\vec{q} = \vec{v}$.
% For the $\vec{q} < \vec{v}$ case, the value of $\unlabelledmast(\vec{q})$ can be assumed 
% to already have been computed at a previous step.  The case $\vec{q} = \vec{v}$
% requires more work.

Suppose $\vec{v}$ is a node vector for $\tau$.
Let $C(\vec{v})$ be the set of all possible child vectors for $\vec{v}$.
Note that since the maximum degree 
of a tree shape in $\tau$ is $d$, $C(\vec{v})$ has at most $d^k$ elements.
The \emph{compatibility graph} $G_{\vec{v}} = (V, E, w)$ for $\vec{v}$ 
is a weighted graph with 
vertex set $V = C(\vec{v})$ 
and edge set $E = \{\vec{v_1}\vec{v_2} : \vec{v_1}$ and $\vec{v_2}$ 
are compatible$\}$.
Each vertex $\vec{v'}$ is weighted by 
$w(\vec{v'}) = \unlabelledmast(\vec{v'})$. We
denote by $MWC(G_{\vec{v}})$ the maximum weight of a clique in $G_{\vec{v}}$.
Moreover, we put $bestDP(\vec{v}) = 
\max_{\vec{u} \in DP(\vec{v})}(\unlabelledmast(\vec{u}))$ in case
$DP(\vec{v})\not=\emptyset$.

We can now state our dynamic programming recurrence.

\begin{lemma}\label{lemma6}
Suppose $\tau=\{T_1,\ldots, T_k\}$ is a set of $k$ tree shapes 
such that each node 
of a tree shape has at most $d$ children. Then
$$
\unlabelledmast(\vec{v}) = \begin{cases}
1 & \text{ if } v_i\in L(T_i), \text{ for all } i\in [k];\\
%\text{if every element of $\vec{v}$ is a leaf of $\tau$} \\
\max (bestDP(\vec{v}), MWC(G_{\vec{v}})) & \text{otherwise.}
\end{cases}
$$
%holds for any node vector $\vec{v}$ of a tree shape in $\tau$.
\end{lemma}

\begin{proof}
%The proof is by induction over the node vector ordering given
%by ``$<$''.  
In the case when all members of $\vec{v}$ are leaves, the lemma  is easily verified.
Suppose that some member of $\vec{v}$ is not a leaf.
%For the induction step, suppose the statement is true for 
%each $\vec{u} < \vec{v}$, 
Let $K := \max (bestDP(\vec{v}), MWC(G_{\vec{v}}))$.  
We first show that $\unlabelledmast(\vec{v}) \leq K$. 
Label the leaves of $T_1(v_1), \ldots, T_k(v_k)$ 
in terms of a set $X$ 
so that the resulting set $\tau'$ of rooted partial $X$-trees have 
a $MAST$ $M$ of maximum size. 
For all $i\in[k]$, let  $f_i:V(M)\to V(T_i)$
be the injection that witnesses that $M$ is a $MAST$ for $\tau'$.
Let $\vec{q} = (f_1(r(M)), \ldots, f_k(r(M))$.  
Note that since only the leaves that descend from a node in $\vec{q}$ 
contribute to $\unlabelledmast(\vec{v})$, it follows that for any $\vec{w}$ such that $\vec{q} \leq \vec{w} \leq \vec{v}$, 
we have $\unlabelledmast(\vec{w}) = \unlabelledmast(\vec{q})$.

If $\vec{q} < \vec{v}$, 
then there must be a direct predecessor $\vec{u}$ of $\vec{v}$ such that 
$\vec{q} \leq \vec{u}$.  In this case, $\unlabelledmast(\vec{v}) = \unlabelledmast(\vec{q}) = \unlabelledmast(\vec{u}) \leq bestDP(\vec{v}) \leq K$.

If $\vec{q} = \vec{v}$, let $z_1, \ldots, z_j$ be the children of $r(M)$.
For each $i \in [j]$, let $\vec{z^i} = (f_1(z_i), \ldots, f_k(z_i))$.
Then although $\vec{z^i}$ need not be a child vector of $\vec{v}$ there
must exist for each $\vec{z^i}$ a child vector 
$\vec{v^i}$ of $\vec{v}$ 
such that $\vec{z^i} \leq \vec{v^i}$. Moreover, since $f_i$ is an
injection for all $i\in [k]$ any two distinct child 
vectors in $\xi=\{\vec{v^1}, \ldots, \vec{v^j}\}$
must be compatible. Hence, the subgraph of $G_{\vec{v}}$ induced by 
$\xi$ is a clique.  
Since the size of $M$ is at most $\sum_{i=1}^j\unlabelledmast(\vec{v^i})$,
 it follows that $\unlabelledmast(\vec{v}) \leq MWC(G_{\vec{v}}) \leq K$.

To see that $\unlabelledmast(\vec{v}) \geq K$
holds too, suppose first that $K = bestDP(\vec{v})$.  
Then $\unlabelledmast(\vec{v}) \geq K$ follows easily, 
as any MAST for a direct predecessor of $\vec{v}$ is also a MAST for $\vec{v}$.
So assume that $K = MWC(G_{\vec{v}})$. Let  
$\vec{v^1}, \ldots, \vec{v^j}$, $j\geq 2$, be the child vectors 
of $\vec{v}$ in a maximum weight clique of $G_{\vec{v}}$.  Consider
some $i\in [j]$. 
Then since $\vec{v^i} < \vec{v}$, we may assign labels to the leaf set of the sub-tree shapes 
$T_i(v^i_l)$, $l\in [k]$, 
of the elements in $\tau$
so that the resulting labelled trees agree on a MAST of 
size $\unlabelledmast(\vec{v^i})$. Since the child vectors 
of $\vec{v}$ are compatible, we may apply this operation to 
each $\vec{v^i}$ vector separately 
so that the resulting labelled trees agree on a a MAST of size at least 
$\sum_{i=1}^j\unlabelledmast(\vec{v^i})  = MWC(G_{\vec{v}})$, 
implying $\unlabelledmast(\vec{v}) \geq K$.
\qed
\end{proof}

\begin{theorem}\label{fptumast}
The unlabelled-MAST problem can be solved in time $O((2n)^k \cdot d^3 \cdot 2^{d \log d \cdot (k + 1/2)})$.
\end{theorem}

\begin{proof}
The algorithm simply traverses all possible node vectors in increasing order w.r.t. $<$, 
then computes $\unlabelledmast(\vec{v})$ in order using Lemma~\ref{lemma6}. 
Consider the time taken for a specific node vector $\vec{v}$, 
assuming that $\unlabelledmast(\vec{u})$ is known for every $\vec{u} < \vec{v}$.
The value $bestDP(\vec{v})$ can be computed in time $O(kd)$.
As for $MWC(\vec{v})$, we first need to construct $G_{\vec{v}}$.
It has at most $d^k$ vertices and $d^{2k}$ edges.  The weights take no time to compute, as they were computed by dynamic programming.  Verifying compatibility of 
two child vectors can be done in time $O(k)$, and so the construction time is $O(kd^{2k})$.
To find $MWC(G_{\vec{v}})$, we can simply enumerate each clique, and since no clique has size more than $d$, 
we can check the subsets of at most $d$ vertices only.  This takes time 
at most 
$\sum_{l =1}^{d} l^2 {{d^k} \choose l} = O(d^3 {{d^k} \choose d} ) = 
O(d^3 (d^k \cdot e / d)^d)$ (using the ${x \choose y} \leq \left( xe/y \right)^y$ inequality).
Since $d \geq 2$,  $e \leq d^{3/2}$, and the above expression is thus 
$O(d^3 d^{d(k + 1/2)}) = O(d^{dk + d/2 + 3})$, 
which dominates the complexity for $\vec{v}$.  
%\kh {I do not quite see this formula. Also what is $e$?}
%\ml{I added some steps.  The initial summation is the time required to check, for every subset of at most 
%$l \leq d$ vertices, whether it has all possible $O(l^2)$ edges to form a complete graph and sum the edge weights
%(this was presented in the same succinct manner in~\cite{farach1995agreement} since this comes up often in computer science, so I didn't feel it required more explanations
%than they did, but please feel free to add more if you think it should be explained more).  
%$e$ is the base of the natural log.}
Since these operations are computed for all the 
$O((2n)^k)$ possible node vectors, the time of the algorithm is 
$O((2n)^k d^{dk + d/2 + 3}) = O((2n)^k d^3 2^{d \log d \cdot (k + 1/2)})$.
\qed
\end{proof}

\section{Open problems}

We conclude with presenting some  open problems:

\begin{itemize}
	\item  Is there is a fixed parameter tractable algorithm
	for computing $d^*_{path}$? 
	\item Is there is a constant factor approximation algorithm for computing $d^*_{path}$
	or $d_{RF}^*$?
	\item Is it NP-hard to compute the extension of the triplet distance \cite{critchlow1996triples} on phylogenetic trees?
	\item Is the computation of $\unlabelledmast(\{T_1, T_2, T_3\})$ NP-hard for $T_1,T_2,T_3$ three 
	binary tree shapes?
\end{itemize}

\subsection*{Funding}

N.E.-M. and M.L. acknowledge the Natural Sciences and Engineering Research Council of Canada (NSERC)
for the financial support of this project.

\bibliographystyle{plain}% the recommended bibstyle
\bibliography{main}      % Bibliography file (usually '*.bib' )

\newpage 

\subsection{Appendix: Hardness of the RF binary case}

We show here that computing $\drf$ is NP-hard even when both tree shapes
are binary.
Some more notation is needed beforehand.
The set of nodes of a tree shape $T$ that have size $k$ is denoted $\setsize{T}{k}$.
If $S$ is a set of integers, $even(S)$ is the set of even integers in $S$, 
whereas $odd(S)$ is the set of odd integers in $S$.
A \emph{cherry} is a the tree shape on $2$ leaves, 
and a \emph{\twoch}~is a tree shape on four leaves that has two cherries.
We denote by $\nbch(T)$ the number of cherries in  a tree shape $T$, and by $\nbtwoch(T)$
the number of \twochs~in $T$.  Note that these are not counted separately:
each \twoch~in a tree increases $\nbch(T)$ by two.
For two tree shapes $T_1$ and $T_2$, \emph{joining} $T_1$ and $T_2$ 
consists in creating a new node $x$ and making $\rt(T_1)$ and $\rt(T_2)$
children of $x$, resulting in a new tree shape $T'$.
We say that we \emph{append} $T_2$ to $T_1$ when joining $T_1$ and $T_2$, 
and letting the resulting tree be the new $T_1$.

In this section, for brevity by a matching we mean a consistent cluster matching.
Given two tree shapes $T_1$ and $T_2$ of the same size,
We first find an upper bound $\ub$ on $\mu(T_1, T_2)$.
Ideally, we would like to be able to find a matching that attains this bound.  However, this is not always feasible and, as we will show, it is NP-hard to decide if there is a matching that attains it.  Let us proceed with the details.

Clearly, if there is a node $u \in V(T_1)$ 
such that $\size(u) \neq \size(v)$ for every $v \in V(T_2)$, 
then $u$ can never be matched in any matching between $T_1$ and $T_2$.
More generally, let $k \in [n]$ and suppose that 
$|\setsize{T_1}{k}| < |\setsize{T_2}{k}|$.  Then a matching $\M$ can contain at most $|\setsize{T_1}{k}|$ pairs of 
nodes of size $k$.

The ideas above imply an upper bound on $\match(T_1, T_2)$ that can be computed easily:
we define 

$$ \ub(T_1, T_2) = \sum_{i = 1}^{n} \min (|\setsize{T_1}{i}|, |\setsize{T_2}{i}|) $$

It is easy to see that $\match(T_1, T_2) \leq \ub(T_1, T_2)$.
We show that deciding if $\mu(T_1, T_2) = \ub(T_1, T_2)$ is NP-hard, by a reduction from the problem DOMINATING SET IN CUBIC GRAPHS.
Given a connected graph $G$ in which every vertex has at most $3$ neighbors and an integer $k$, this problem asks if $G$ contains a dominating set of cardinality at most $k$.  
This problem is known to be NP-hard (and in fact, APX-hard)~\cite{alimonti1997hardness}.

Our reduction constructs two binary tree shapes $T_1$ and $T_2$ from $G$.
Roughly speaking, we would like to re-use the  ideas of Theorem~\ref{theo:np-complete-height-3} for 
the non-binary case.  Ideally, we would be able to simply take the $w_i, d_{i,j}, w'_i$ and $d'_j$ nodes from this construction, 
and simply ``binarize'' these non-binary nodes.  This, however, will create many new internal nodes, 
which we may or may not be able to match.  The hardness proof presented in that theorem relied on the fact that $d'_j$ nodes could only be matched 
with $d_{i,j}$ nodes, and $w'_i$ nodes with $w_i$ nodes.  Maintaining such a
fine-grained control over which 
nodes can be matched together is much more difficult in the binary case.  

We will therefore need to construct sub-tree shapes that contain and avoid prescribed cluster sizes, 
and that have a certain number of cherries and of \twochs, and hence we  develop some constructive 
tools before proceeding with our reduction.
An important idea behind the reduction is that the nodes of $T_1$ will mostly be of odd size and 
those of $T_2$ mostly of even size, with the exception of some \emph{special} node sizes that are common to both trees.
One technical part of the reduction is that both trees must contain cherries and \twochs, 
and we need to keep track of their counts to ensure that the reduction works as desired.

We call a pair of sets of integers $(H, M)$ 
 \emph{well-behaved} if $H \cap M = \emptyset$, $\max (H \cup M) \geq 4$, and 
$|i - j| \geq 4$ for any distinct $i, j \in H \cup M$.  %The \emph{maximum} of $(R, S)$ is $\max (R \cup S)$.
In what follows, $H$ will usually be used for node sizes to ``hit'', and $M$ for
node sizes to ``miss''.

More precisely, we say that a tree shape $T$ \emph{hits} an integer $i$ if 
$T$ has a \emph{unique} node $v$ such that $\size(v) = i$, 
and $T$ \emph{misses} $i$ if 
$\size(v) \neq i$ for every $v \in V(T)$.
For some well-behaved sets $(H, M)$, 
we say that $T$ is an \emph{$(H, M)$-tree  shape} if 
$T$ hits $h$ for every $h \in H$ and misses $m$ for every $m \in M$.

An $(H, M)$-tree shape $T$ is \emph{odd} if, in addition to the above, it also 
satisfies the condition that 
for every $v \in V(T)$ such that $\size(v) \notin H$ and $\size(v) > 4$, the size of $v$ is odd.
Likewise, $T$ is \emph{even} if for every $v \in V(T)$ such that $\size(v) \notin H$ and $\size(v) > 4$, 
the size of $v$ is even.
The next technical lemma allows us to construct $(H, M)$-tree shapes while having some control over the cherries and \twochs.
The particular conditions of the lemma will all be of use later on.

\begin{lemma}\label{lem:even-odd-trees}
Let $(H, M)$ be well-behaved sets such that $c := \max(H) > \max(M)$ is an odd integer,
and let $q$ be an integer with $|M| \leq q \leq \max (|M|, c/4 - 5(|H| + |M| + 1))$.

Then if $odd(M) = \emptyset$, there exists an even $(H, M)$-tree shape $T_e$ of size $c$ with 
$(c + 1)/2 - |odd(H)|$ cherries and $q$ \twochs.
Similarly, if $even(M) = \emptyset$, there exists an odd $(H, M)$-tree shape $T_o$ of size $c$ with $(c - 1)/2 - |even(H)|$ cherries and $q$ \twochs.  
Moreover, both $T_e$ and $T_o$ can be constructed in polynomial time.
\end{lemma}

\begin{proof}
We make the construction explicit in
Algorithm~\ref{algo:todd}, which shows how to construct $T_o$ and $T_e$ 
(if the $isOdd$ input is true, $T_o$ is built, and otherwise $T_e$ is built).
Roughly speaking, we start with $T$ a tree shape on two leaves for $T_e$ and 
three leaves for $T_o$. 
While $T$ does not have size $c$ or $c - 1$, 
we join $T$ with either a cherry or a \twoch~in order to maintain the odd/even parity requirement 
of the node sizes.
Exceptionally, we may append a single leaf when an element of $H$ requires us to hit a particular size, 
and this is followed by appending another single leaf to return to the desired parity.
The values of $M$ are skipped by 
appending a \twoch, enforcing at least $|M|$ \twochs~in $T$.
The remaining $q - |M|$ \twochs~are appended
whenever $H$ and $M$ allow it during the construction.
See Algorithm~\ref{algo:todd} for details.

\begin{algorithm}
\caption{Algorithm to construct $T_e$ or $T_o$}\label{algo:todd}
\begin{algorithmic}[1]
\Procedure{buildTree}{$H, M, q, isOdd$}
%\Statex If isOdd is true, we build $T_o$, otherwise we build $T_e$

\State Let $T$ be a cherry
\State \algorithmicif\ isOdd = True \algorithmicthen ~Append a single leaf to $T$ 
%\State $s \gets 3$ \Comment{$s$ will contain size of $T$ during the construction}

%\State \algorithmicif\ isOdd = True \algorithmicthen\ $nbSkips = |odd(M)|$ \algorithmicelse\ $nbSkips = |even(M)|$
\State $r \gets q - |M|$ \Comment{$r$ is the number of \twochs~that are not enforced by $M$}

\While{$\size(T) \neq c$} \label{line:main-loop-appendix}
  \State Let $s := \size(T)$
  \If{$s + 1 = c$}
  	 \State Append a single leaf to $T$  
  \ElsIf{$s + 1 \in H$}
      \State Append a single leaf to $T$, then append another single leaf to $T$
      %\State Append a single leaf to $T$
      %\State Increase $s$ by $2$
  \ElsIf{$s + 2 \in H$}
      \State Append a cherry to $T$
      %\State $s \gets s + 2$
  \ElsIf{$s + 2 \in M$ \label{line:s2inm}}
      \State Append a \twoch~to $T$ \label{line:s2inm_append}
      %\State $s \gets s + 4$
  \ElsIf{$s + 3 \in H$ or $s + 4 \in M$  \label{line:cantdbl}}
      \State Append a cherry to $T$
      %\State $s \gets s + 2$
  \Else \label{line:last_else}
      \If{$r > 0$}
          \State Append a \twoch \label{line:dblch_other} to $T$ and decrease $r$ by $1$
          %\State $s \gets s + 4$
          %\State $r \gets r - 1$
          %\State Decrease $r$ by $1$
      \Else
          \State Append a cherry to $T$
          %\State $s \gets s + 2$
      \EndIf
  \EndIf
\EndWhile
\Return $T$
\EndProcedure
\end{algorithmic}
\end{algorithm}

Clearly, Algorithm~\ref{algo:todd} takes polynomial time.
We now show the correctness of the procedure, i.e. that the output tree $T$ satisfies the conditions of $T_e$ and $T_o$.
%First observe that at every iteration of the main \textbf{while} loop, the size of $T$ increases by $2$ or $4$, and that 
%the value of $s$ is always odd.
Let $w$ be the number of times the algorithm enters the ``while''
loop on line~\ref{line:main-loop-appendix}, 
and for $i \in [w]$, let $s_i$ be the value of $s$ at the start of 
the $i$-th iteration.
The lemma's statement is easy to verify if $w = 1$, so we assume $w > 1$.
Denote $s_{w + 1} := \size(T)$, i.e. the final size of $T$, 
and let $S = \{s_1, \ldots, s_w, s_{w + 1}\}$.
We have $s_1 = 2$ for $T_e$ and $s_1 = 3$ for $T_o$
$s_{i + 1} - s_i \in \{2, 4\}$ for any $i \in [w - 1]$.  
Moreover, each $s_i \in S$ must be even for $T_e$ (except $s_{w + 1}$) and odd for $T_o$.

We first show that $T$ hits every $h \in H$.
Suppose instead that $T$ does not hit some $h \in H$.
Let $i \in [w]$ such that
$s_i < h$ and $s_{i + 1} > h$.
Then $s_i < h < s_i + 4$, since $s_{i + 1} \leq s_i + 4$.
But all cases $s_i + 1, s_i + 2, s_i + 3 \in H$ are explicitly checked
by the algorithm: in the first two cases, $h$ gets hit by $T$, 
and in the $s_i + 3 \in H$ case, $s_{i + 1} = s_i + 2 < h$, contradicting $s_{i + 1} > h$
(note that the `else if' on line~\ref{line:s2inm} cannot be entered 
since $(H, M)$ is well-behaved).
Observe that since $c \in H$, this implies that $\size(T) = c$ as desired.

We next show that $T$ misses every size in $M$.
Let $m \in M$, and let $i \in [w]$ 
such that $s_i < m \leq s_{i + 1}$.
In the case of $T_e$ (resp. $T_o$), $m$ must be even (resp. odd), 
since 
by assumption we have $odd(M) = \emptyset$ (resp. $even(M) = \emptyset$).  
As $s_i$ has the same parity as $m$, we must have
$m = s_i + 2$ or $m = s_i + 4$.
If $m = s_i + 2$, line~\ref{line:s2inm} ensures that $T$ misses $m$
(no other ``if'' case can be entered by the well-behaved property). 
If $m = s_i + 4$, none of the cases that append a \twoch~apply
(lines~\ref{line:s2inm} and~\ref{line:last_else}), 
and so $s_{i + 1} \leq s_i + 2 = m - 2$, contradicting our assumption that $s_{i + 1} \geq m$.

Next, we show that for the $T_e$ case, $\nbch(T) = (c + 1)/2 - |odd(H)|$.
If suffices to observe that in general, if a binary rooted tree shape $T'$ contains
$l$ leaves that do not belong to a cherry, which we will call \emph{single} leaves, then $\nbch(T') = (\size(T') - l)/2$.
When constructing $T_e$, the algorithm appends two single leaves for each element of $odd(H) \setminus \{c\}$, 
plus a single leaf at the step when the size of $T$ is $c - 1$. 
The number of appended single leaves is therefore $2(|odd(H)| - 1) + 1$.
All other cases append a cherry or a \twoch, and so 
$\nbch(T_e) = (c - (2|odd(H)| - 1))/2 = (c + 1)/2 - |odd(H)|$, as desired. 
For the $T_o$ case, two single leaves are appended for each element of $even(H)$
(notice that $c$ is odd, and so $c \notin even(H)$)
and a single leaf is appended before the \textbf{while} loop, 
leading to $\nbch(T_o) = (c - (2|even(H)| + 1))/2 = (c - 1)/2 - |even(H)|$ in the same manner.

It only remains to show that $\nbtwoch(T) = q$.  
%Consider the $T_e$ case
%(the $T_o$ case can be seen to be correct as well by replacing 
%$even(M)$ by $odd(M)$ in what follows).
First note that in both the $T_e$ and $T_o$ cases, 
the ``if'' case on line~\ref{line:s2inm} is entered 
exactly $|M|$ times, and so $|M|$ \twochs~are due to line~\ref{line:s2inm_append}.
If $q = |M|$, we are done, so assume $q > |M|$.
Let us count the number of times $t$ that the
``else'' statement is entered on line~\ref{line:last_else}.
We show that $t \geq c/4 - 5(|H| + |M| + 1)$, which proves the desired result
since $r$ ensures that the exact number of \twochs~are appended.
Observe first that $s_1 \geq 3$, $s_{|S|} = c$ and $s_i \leq s_{i - 1} + 4$ for each $2 \leq i \leq |S|$.
Hence, $|S| \geq \lfloor (c - 3)/4 \rfloor$.
Also, the case on line~\ref{line:last_else} is entered for each 
$s_i \in S$ such that $\{s_i, s_i + 1, \ldots, s_i + 4\} \cap H \cup M = \emptyset$.
For each $h \in H \cup M$, there are at most $5$ members of $s_i \in S$ such that $s_i + j \in H \cup M$, $0 \leq j \leq 4$.
Therefore, the number of members of $S$ that do satisfy the condition for entering line~\ref{line:last_else}
is at least $|S| - 5(|H| + |M|) \geq  \lfloor (c - 3)/4 \rfloor - 5(|H| + |M|) \geq 
c/4 - 5(|H| + |M| + 1)$. %\footnote{This bound is certainly not tight, but we will not bother here.}
We have verified every required property for $T_e$ and $T_o$, concluding the proof.
\qed 
\end{proof}

%\vspace{2mm}

Note that in the tree shape $T$ constructed in Lemma~\ref{lem:even-odd-trees}, 
every node of size $4$ is the root of a \twoch, unless 
$T = T_o$ and $4 \in H$.  The reader should bear in mind that this
particular case will never occur, and we will assume that in what follows, every subtree of size $4$ is a \twoch.

We are now ready to describe our reduction.
Let $(G, k)$ be an instance of DOMINATING SET IN CUBIC GRAPHS, and let $(v_1, \ldots, v_n)$ 
be an (arbitrary) ordering of $V(G)$.
Since a vertex can only dominate four vertices (its neighbors plus itself), we may 
assume that $k \geq n/4$.

For each $i \in [n]$, let $r_i = 4i + 1$ and $w_i = 80n + 4i + 5$.  
We also let $S = \{r_1, \ldots, r_n, w_1, \ldots, w_n\}$, and call an element $s \in S$ a \emph{special size}.
Moreover, let $\hat{r} = 20n + 1$.
For $i \in [n]$, put $R_i = \{r_1, r_2, \ldots, r_i\}$
and $R_0 = \emptyset$.
Note that for each $i \in [n]$, $(\{r_i\}, R_{i - 1})$ is well-behaved.
Our goal is to construct $T_1$ and $T_2$ from $(G, k)$ such that 
the only possible nodes that can be matched are either cherries, \twochs, or have a special size.
We will have, with hindsight, $\size(T_1) = \size(T_2) = 200n^8 + 100n^4 + 82n^2 + 7n + k(84n + 5)$.
We start by constructing the $T_1$ tree shape, for which we need to define five types of sub-tree shapes, as follows.
%: (1) the $D_i$ tree shapes for $i \in [n]$, which have special size $r_i$;
%(2) the $\hat{D_i}$ tree shapes, which all have the same size $20n + 1$ and contain $D_i$ as a sub-tree shape;
%(3) the $W_i$ tree shapes, which have special size $w_i$ and contain exactly four $\hat{D_j}$ trees as sub-tree shapes (for four distinct $j$'s), one for each neighbor $v_j$ in the closed neighborhood of $v_i$ in $G$;
%(4) the $U^i$ tree shapes for $i \in [k]$, which hit every size $w_j$ and miss every size $r_j$;
%(5) the $B_1$ and $B_2$ sub-tree shapes, 
%respectively of size $100n^4$ and $200n^8$, which miss every special size and are used to control the size of the nodes above the $W_i$ and $U^i$ trees.  We describe these trees in %detail.  We will need to keep track of the number of cherries and \twochs, since their counts has to match in certain pairs of subtrees from $T_1$ and $T_2$.
Figure~\ref{fig:di_and_wi} provides an illustration of the $D_i$ and $W_i$ trees, and Figure~\ref{fig:t1_t2} of the $T_1$ and $T_2$ trees.

\begin{figure*}[!b]
\begin{center}
\includegraphics[width= 0.95 \textwidth]{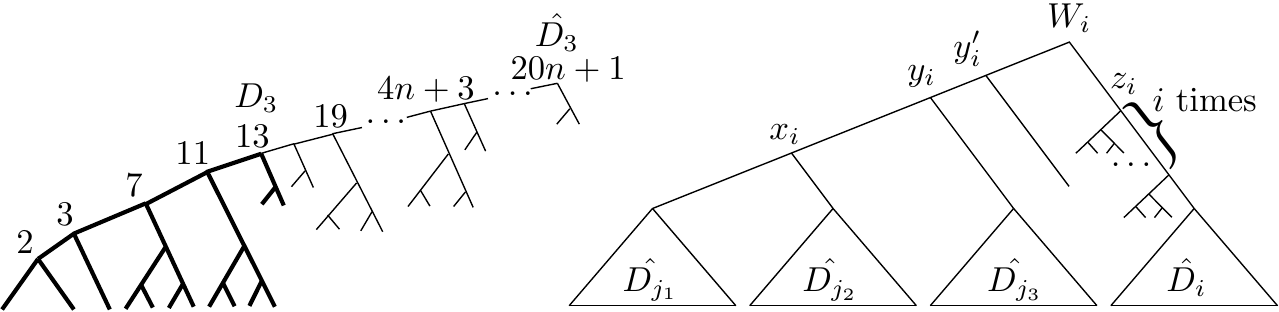}
\caption{On the left, an illustration of $D_i$ (with heavy edges) and $\hat{D_i}$ with $i = 3$.  The labels correspond to the sizes of the internal nodes.  Here, $D_3$ is an odd $(\{13\}, \{5, 9\})$-tree shape.  Then $\hat{D_3}$ is obtained from $D_3$ by first appending a cherry, then $n - 3$ double-cherries, and finally $8n - 1$ cherries. Note that $\hat{D_3}$ hits $r_3 = 13$ and misses $\{5, 9, 17, 21, \ldots, 4n + 1\}$.  On the right, an illustration of a $W_i$ tree.}\label{fig:di_and_wi}
\end{center}
\end{figure*}

\vspace{2mm}
\noindent
\textbf{The $D_i$ and $\hat{D_i}$ tree shapes:} 
For each $i \in [n]$, let $D_i$ be an odd $(\{r_i\}, R_{i - 1})$-tree shape
of size $r_i = 4i + 1$ with $\nbch(D_i) = (r_i - 1)/2 = 2i$ and $\nbtwoch(D_i) = |R_{i - 1}| = i - 1$
(which can be constructed in polynomial time, by Lemma~\ref{lem:even-odd-trees}).

Then, obtain the tree shape $\hat{D_i}$ of size $\hat{r}$ from $D_i$ by first joining $D_i$ with a cherry, 
then successively appending a \twoch~$n - i$ times, resulting in a tree shape of size $4i + 1 + 2 + 4(n - i) = 4n + 3$.  Then join this tree shape with a cherry 
$8n - 1$ times, resulting in $\hat{D_i}$ (of size $20n + 1 = \hat{r}$).
See Figure~\ref{fig:di_and_wi}.
Note that $\hat{D_i}$ is an odd $(\{r_i\}, R_n \setminus \{r_i\})$-tree shape of size $\hat{r}$.
Also, $\nbch(\hat{D_i}) = 10n$ and $\nbtwoch(\hat{D_i}) = i - 1 + (n - i) = n - 1$.

\vspace{2mm}
\noindent
\textbf{The $W_i$ tree shapes:}
for each $i \in [n]$, we construct the tree shape $W_i$ as follows.
Let $v_{j_1}, v_{j_2}, v_{j_3}$ be the neighbors of $v_i$ in $G$.
First take a copy of $\hat{D}_{j_1}$ and a copy of $\hat{D}_{j_2}$, 
and join their roots under a common parent $x_i$.
Take a copy of $\hat{D}_{j_3}$ and join its root and $x_i$ under a common 
parent $y_i$.
Then join $y_i$ and a single leaf under a common parent $y_i'$.
Finally take a copy of $\hat{D}_i$, append a \twoch~$i$ times to it
to obtain a new tree shape rooted at a node $z_i$, 
and connect $y_i'$ and $z_i$ under a common parent, which is the root of $W_i$ (see Figure~\ref{fig:di_and_wi}).

For later use, we will need two sizes that are missed by every $W_i$.
Let $w^* := 3\hat{r} + 2$ and $w^{**} := 3\hat{r} + 4$.
It is straightforward to verify that $W_i$ misses $w^*$ and $w^{**}$.

The following properties hold for each $W_i$ tree:

\begin{enumerate}

\item
$\size(W_i) = 4 \hat{r} + 4i + 1 = 80n + 4i + 5 = w_i$;

\item
$\size(x_i) = 2\hat{r},  \size(y_i) = 3 \hat{r}, \size(y'_i) = 3 \hat{r} + 1$ and 
$\size(z_i) = \hat{r} + 4i$.
Moreover, $x_i$ and $y'_i$ are the only two nodes of even size in $W_i$, except cherries and \twochs;

\item
$\nbch(W_i) = 4 \cdot 10n + 2i = 40n + 2i$;

\item
$\nbtwoch(W_i) = 4 (n - 1) + i = 4n + i - 4$;

\item
$W_i$ hits $r_{j_1}, r_{j_2}, r_{j_3}, r_i$ and misses every other element of $R_n$;

%\item
%$W_i$ misses $w^*$ and $w^{**}$.

\end{enumerate}

%The idea behind the $W_i$ tree is that its root  can either be matched with a cluster of $T_2$
%or not.  If so, then the $D_j$ subtrees under $W_i$ will become unavailable, and if it 
%is not matched, its $D_{j}$ subtrees can all be matched.
%Intuitively, matching the $D_j$ subtree under $W_i$ represents $v_j$ being dominated by $v_i$.
%If $D_j$ is matched however, there needs to be an alternative for the special size $w_i$.
%We will hence build a tree $U$ that contains, for each $i \in [n]$, a subtree 
%with the same properties as $W_i$ (i.e. size, cherries and \twochs).
%The main difference is that $U$ misses every integer of $R_n$.
%We will then need $k$ copies $U^1, \ldots, U^k$ of $U$ for the correspondence with a dominating set of size $k$
%in $G$.

\vspace{2mm}
\noindent
\textbf{The $U$ sub-tree shape: }
we construct a tree shape $U$ inductively as follows:
first, $U_1$ is an odd $( \{2\hat{r}, 3\hat{r} + 1\}, R_n )$-tree
of size $\size(W_1) = 80n + 9$ with $\nbch(U_1) = \nbch(W_1) = 40n + 2$ and $\nbtwoch(U_1) =  \nbtwoch(W_1) = 4n - 3$.
Notice that $U_1$ has the same two even node sizes as every $W_i$ tree shape.
One can check that $U_1$ satisfies all the conditions of Lemma~\ref{lem:even-odd-trees}, and hence can be constructed.  
%In particular, for the number of \twochs, the condition on $q = 4n - 3$ from the Lemma 
%is satisfied since 
%$|R_n| = n \leq 4n - 3 \leq \size(U_1)/4 - 6 \cdot (n + 4)$. 

Then, for $1 < i \leq n$, 
$U_i$ is obtained by joining $U_{i - 1}$ with a \twoch.
Since $\size(W_i) = \size(W_{i - 1}) + 4$, $\nbch(W_i) = \nbch(W_{i - 1}) + 2$ and $\nbtwoch(W_i) = \nbtwoch(W_{i - 1}) + 1$, 
$U_i$ has the same size, number of cherries and \twochs~as $W_i$.
We let $U = U_n$, noting that $\size(U) = \size(W_n)$.

\vspace{2mm}
\noindent
\textbf{The $B_1$ and $B_2$ sub-tree shapes: }
the $B_1$ (respectively, $B_2$) tree shape is simply an odd 
$(\emptyset, S)$-tree of size $100n^4$ (resp. of size $200n^8$).  We are not concerned with their number of cherries or \twochs.

\begin{figure*}[!t]
\begin{center}
\includegraphics[width= 1 \textwidth]{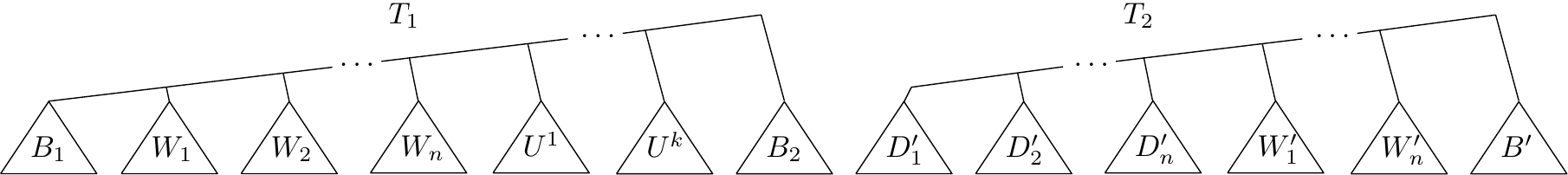}
\caption{The $T_1$ and $T_2$ tree shapes.}\label{fig:t1_t2}
\end{center}
\end{figure*}

Finally, $T_1$ is obtained by taking a caterpillar shape
on $n + k + 2$ leaves
%$\{\l_1, \ldots, \l, \l_{n + 1}, \l_{n + 2}, \ldots, \l_{n + k}\}$, 
$\{\l_1, \ldots, \l_{n + k + 2}\}$,
and replacing $\l_1$ by the $B_1$ tree shape, 
replacing each $\l_{i + 1}$ by the $W_i$ tree shape for $i \in [n]$, 
each $\l_{n + 1 + j}$ by a copy $U^j$ of $U$ for $j \in [k]$,
and $\l_{n + k + 2}$ by the $B_2$ sub-tree shape (see Figure~\ref{fig:t1_t2}).
We have 
$\size(T_1) = 200n^8 + 100n^4 + k \cdot \size(U) + \sum_{i = 1}^n \size(W_i) = 200n^8 + 100n^4 + 82 n^2 + 7n + k (84 n + 5)$, as predicted.
Denote by $I$ the set of even sizes of nodes in $T_1$ greater than $4$, i.e. $I = \{\size(u) : u \in V(T), \size(u) > 4$ and $\size(u)$ is even$\}$. 
Note that the $W_i$ and $U$ tree shapes, together, contribute to only two sizes in $I$ (namely $2\hat{r}$ and $3\hat{r} + 1$), and the only other even size nodes must be ancestors of $\rt(B_1)$. 
Hence if $i \in I \setminus \{2\hat{r}, 3\hat{r} + 1\}$, 
then $i = \theta(n^4)$ (unless $i = \size(\rt(T_1))$, in which case $i = \theta(n^8)$).
Moreover, it is not hard to see that $|i_1 - i_2| \geq 4$ for any 
distinct $i_1, i_2 \in I$.

%the properties of $T_1$ are as follows: 
%\begin{align*}
%\size(T_1) &= k\size(U) + \sum_{i = 1}^n \size(W_i) = 82 n^2 + k (84 n + k) + 7 n \\
%\nbch(T_1) &= k\nbch(U) + \sum_{i = 1}^n \nbch(W_i) = 41 n^2 + 42 k n +  n \\
%\nbtwoch(T_1) &= k\nbtwoch(U) + \sum_{i = 1}^n \nbtwoch(W_i)  = (9 n^2)/2 + k (5 n - 4) - (7 n)/2
%\end{align*}

Now we may construct $T_2$.  %Our goal is that $T_2$ has %enough cherries and \twochs~that can be made in common in %$T_1$, plus 
%$2n$ incomparable
%clusters, one for each special size.
Recall that $S = \{r_1, \ldots, r_n, w_1, \ldots, w_n\}$ is the set of special sizes.  
For each $s \in S$, we want $T_2$ to contain exactly one node of size $s$, 
such that any pair of nodes having a special size are incomparable. 

For each $i \in [n]$, let $R'_i = \{4, 8, \ldots, 4i\}$ and $R'_0 = \emptyset$.  Then let $D'_i$ be 
an even $(\{r_i\}, R'_{i - 1})$-tree of size $r_i$ with 
$\nbch(D'_i) = \nbch(D_i) = 2i$ and $\nbtwoch(D'_i) = i - 1$
(Lemma~\ref{lem:even-odd-trees} ensures that $D'_i$ can be constructed).  
Then let $W'_i$ be an even $(\{w_i, w^*, w^{**}\}, \{2\hat{r}, 3\hat{r} + 1\}$)-tree
of size $ \size(W_i) = w_i = 80n + 4i + 5$ with 
$\nbch(W'_i) = \nbch(W_i) = 40n + 2i$ and $\nbtwoch(W'_i) = \nbtwoch(W_i) = 4n + i - 4$.
Again, we invoke Lemma~\ref{lem:even-odd-trees} for $W'_i$ (which was in fact
the sole purpose of $w^*$ and $w^{**}$, as they control $\nbch(W'_i)$).
We next build a sub-tree shape $B'$ so that $T_2$ attains the same size 
as $T_1$.
%, along with enough cherries and \twochs~to match against.
That is, let $s_{B'} := \size(T_1) - \sum_{i = 1}^n (\size(D'_i) + \size(W'_i))$
%= 
%82 n^2 + k (84 n + 5) + 7 n - \sum_{i = 1}^n (1 + 4i + 80n + 4i + 5) = k (84 n + 5) - 2n^2 - 3n$. 
%Note that since $k \geq n/4$, this implies $s_F > 0$ and 
%so both trees can be made of the same size.
(observe that $s_{B'}$ is clearly above $0$).
Letting $I' = \{i \in I : i < s_{B'}\}$, 
we let $B'$ be an even $(\{s_{B'}\}, I')$-tree 
of size $s_{B'}$ (to see that $(\{s_{B'}\}, I')$ 
is well-behaved, we have argued that elements of $I$ differ by at least $4$, and for $s_{B'}$, we note that $s_{B'} = \theta(n^8)$ 
whereas $i = \theta(n^4)$ for each $i \in I'$).
The tree shape $T_2$ is obtained by taking a caterpillar shape
on $2n + 1$ leaves $\{\l'_1, \ldots, \l'_{2n + 1}\}$,
replacing $\l'_i$ by the $D'_i$ sub-tree shape and  
 $\l'_{n + i}$ by the $W'_i$ sub-tree shape for each $i \in [n]$, 
 then replacing $\l_{2n + 1}$ by the $B'$ sub-tree shape.
 
 We are finally done with the construction.  First, we calculate the 
 upper bound on $\match(T_1, T_2)$.

\begin{lemma}\label{lem:matching-ub}
$\ub(T_1, T_2) = \size(T_1) + 2n + \min(\nbch(T_1), \nbch(T_2)) + \min(\nbtwoch(T_1), \nbtwoch(T_2))$.
\end{lemma}

\begin{proof}
The $\size(T_1)$ term in $\ub$ is due to the fact that each leaf of $T_1$ can be matched.
Now, by construction, there are $2n$ special sizes in $S$ and for each $s \in S$, 
$|V_s(T_2)| = 1$ and $|V_s(T_1)| \geq 1$.
This implies 

\begin{align*}
\ub(T_1, T_2) &=  \sum_{i = 1}^{n} \min (|\setsize{T_1}{i}|, |\setsize{T_2}{i}|) \\
&\geq \size(T_1) + 2n + \min(\nbch(T_1), \nbch(T_2)) + \min(\nbtwoch(T_1),\nbtwoch(T_2))
\end{align*}
To see that this is also an upper bound on $\ub(T_1, T_2)$, we must argue 
that $T_1$ and $T_2$ share no nodes of the same size greater than $1$, 
except the cherries, \twochs~and the nodes having a special size.  First observe that every node of size $4$ in $T_1$ and $T_2$ is the root of a \twoch.  Thus we may restrict our attention to the sizes greater than $4$.
The $\hat{D_i}, W_i, U^i, B_1$ and $B_2$ sub-tree shapes of $T_1$ are odd, whereas the $D'_i, W'_i$ and $B'$ sub-tree shapes of $T_2$ are even, and are constructed so that the sets of sizes of their nodes intersect only on $S$.
Therefore, if $\size(x_1) = \size(x_2)$ for some non-root $x_1 \in V(T_1)$ 
and non-root $x_2 \in V(T_2)$, $x_1$ must be an ancestor of 
$\rt(B_1)$ and $x_2$ an ancestor of $\rt(D'_1)$.
But by the placement of $B_1$ and $B'$, every ancestor of $r(B_1)$ has size at least $100n^4$, whereas the non-root ancestors of $r(D'_1)$ have size $O(n^2)$.  Hence $\size(x_1) = \size(x_2)$ is not possible.
\qed
\end{proof}

\begin{theorem}
Deciding if $\match(T_1, T_2) = \ub(T_1, T_2)$ is NP-complete.
\end{theorem}

\begin{proof}
The problem is in NP, since $\ub$ is easy to compute
and a matching can be provided as a certificate, which can easily be verified in polynomial time.
As for hardness, let $(G, k)$ be an instance of dominating set on cubic graphs, and 
let $T_1$ and $T_2$ be the corresponding tree shapes constructed as above.
We show that $G$ has a dominating set of size $k$
if and only if $\match(T_1, T_2) = \ub(T_1, T_2)$.

($\Rightarrow$): Let $X = \{v_{d_1}, \ldots, v_{d_k}\}$ be a dominating set of $G$ (we may assume that $|X| = k$, as if $|X| < k$ we may add arbitrary vertices into $X$).  We construct a matching $\M$ of the desired size, the main idea being 
to match the $W'_i$ and $D'_i$ roots to a set of incomparable nodes in $T_1$.
For each $v_{d_i} \in X$, $i \in [k]$, match the $\rt(W'_{d_i})$ node of $T_2$ 
with the root of the $U^i_{d_i}$ sub-tree shape in $T_1$, which is the unique sub-tree shape of size $w_{d_i}$ in the $U^{i}$ sub-tree shape of $T_1$ (since there  are $k$ copies of $U$, each $v_{d_i}$ can be matched in this manner).
Note that by construction, $\nbch(U^i_{d_i}) = \nbch(W'_{d_i})$ and $\nbtwoch(U^i_{d_i}) = \nbtwoch(W'_{d_i})$.  It is straightforward to see that the cherries and \twochs~of the two sub-tree shapes can all be matched in a consistent manner. % (e.g. by finding a bijection between the \twochs, and a bijection between the cherries that do not belong to a \twoch).
Then for each $v_i \notin X$, we match $\rt(W'_i)$ with $\rt(W_i)$.  As before, $\nbch(W_i) = \nbch(W'_i)$ and 
$\nbtwoch(W_i) = \nbtwoch(W'_i)$, and so all 
the cherries and \twochs~of the two sub-tree shapes can be matched.  
So far, all the $W'_i$ roots are matched in a consistent manner, and the nodes $\{\rt(W_i) : v_i \in X\}$ of $T_1$
are unmatched, leaving their $D_j$ sub-tree shapes available.
Thus, for each $v_j \in V(G)$, 
let $v_h \in X$ be a vertex dominating $v_j$
(if $v_j \in X$, it dominates itself and we let $v_j = v_h$).
We match $\rt(D'_j)$ with the root of the $D_j$ sub-tree shape of
$T_1$ that lies within the $W_h$ tree shape.
Once again, $\nbch(D_j) = \nbch(D'_j)$ and 
$\nbtwoch(D_j) = \nbtwoch(D'_j)$, and their cherries/\twochs~can all be matched.
So far, all nodes of $T_1$ that have been matched are incomparable, ensuring consistency.

In order to attain $\ub(T_1, T_2)$, it only remains to match the roots of cherries and \twochs~that do not lie under a matched node.
For these, it is not hard to see that we can compute a maximum matching between the \twochs~first 
(and match their descending cherries together),
then a maximum matching between the cherries that have not been matched in the preceding step.
In this manner, we find a cluster matching of size $\size(T_1) + 2n + \min(\nbch(T_1), \nbch(T_2)) + \min(\nbtwoch(T_1), \nbtwoch(T_2))$.

($\Leftarrow$): Suppose that $\match(T_1, T_2) = \ub(T_1, T_2)$, and let $\M$ be a matching between $T_1$ and $T_2$.
It follows from Lemma~\ref{lem:matching-ub} that every root of the 
$D'_i$ and $W'_i$ sub-tree shapes of $T_2$ must be matched, as matching the special size nodes of $T_2$ is necessary to attain $\ub(T_1, T_2)$.
Each $D'_i$ root can only be matched
with a $D_i$ root in $T_1$, with $\rt(D_i)$ being a descendant of $r(W_j)$ for some $j \in [n]$.
For $i \in [n]$, let $d_i$ be the node of $T_1$ matched with $\rt(D'_i)$ in $\M$, and let $w(D'_i)$ be the unique index $j$ such that $\rt(W_j)$ is an ancestor of $d_i$ in $T_1$.
Let $X = \{v_j \in V(G) : j = w(D'_i)$ for some $i \in [n]\}$.
We claim that $X$ is a dominating set of size at most $k$.
Since, for $i \in [n]$, 
$\rt(W_j)$ is an ancestor of $d_i$ if and only if $v_iv_j \in E(G)$ or $i = j$, 
each $v_i \in V(G)$ must have a neighbor in $X$ or must itself be in $X$, and thus $X$ is indeed a
dominating set.
Now suppose that $|X| > k$.
Note that if some $D'_i$ root is matched with $d_i$ with ancestor
$\rt(W_j)$, then 
$W'_j$ cannot be matched with $W_j$ 
(because $\rt(D'_i)$ and $\rt(W'_j)$ are incomparable, 
whereas $d_i$ and $\rt(W_j)$ are not).
Then the only matching options for the root of $W'_j$ are in the $U^h$ sub-tree shapes 
in $T_1$, $h \in [k]$.  
In other words, for each $v_i \in X$, $r(W'_i)$ is matched with a node in the $U^h$ copy for some $h \in [k]$.
Since $|X| > k$, there must be two distinct $v_{i_1}, v_{i_2} \in X$ such that
both $r(W'_{i_1})$ and $r(W'_{i_2})$ are matched with a sub-tree shape of the same $U^h$ copy.
But this is not possible since the roots of these sub-tree shapes are not incomparable.
This shows that $|X| \leq k$, concluding the proof.
\qed
\end{proof}

\end{document}